\documentclass[draftcls,onecolumn,12pt]{IEEEtran}

\usepackage{amsbsy}
\usepackage{amsmath}
\usepackage{amssymb}
\usepackage{cite}
\usepackage{graphicx}
\usepackage{pstricks}
\usepackage{pst-all}
\usepackage{theorem}
\newcommand{\insertfig}[3]{
\begin{figure}[!tb]\centering
\includegraphics[width=#1\columnwidth]{#2.eps}
\caption{\em #3}\label{#2}\end{figure}}

\newtheorem{intexp}{Integral}[subsection]

\newtheorem{prop}{Property}[subsection]

\newtheorem{theorem}{Theorem}[section]
\theorembodyfont{\rmfamily}

\newtheorem{remark}{Remark}[section]
\def\ben{\begin{enumerate}}
\def\beq{\begin{equation}}
\def\beqa{\begin{eqnarray}}
\def\bit{\begin{itemize}}
\def\een{\end{enumerate}}
\def\eeq{\end{equation}}
\def\eeqa{\end{eqnarray}}
\def\eit{\end{itemize}}
\def\dst{\displaystyle}
\def\non{\nonumber\\}

\DeclareMathAlphabet{\mathsfbf}{OT1}{cmss}{sbc}{n}

\def\Ac{\check{\Am}_1}
\def\Am{\mathbf{A}}
\def\asol{\check{\alpha}}

\def\Bc{\check{\Bm}_1}
\def\Bm{\mathbf{B}}
\def\bsol{\check{\beta}}
\def\bt{\tilde{b}}
\def\bv{\mathbf{b}}
\def\bvt{\widetilde{\bv}}
\def\C{\mathbb{C}}
\def\Cc{\check{\Cm}_1}
\def\Cm{\mathbf{C}}

\def\Dc{\check{\Dm}_1}
\def\Dm{\mathbf{D}}

\def\dt{\tilde{d}}
\def\dv{\mathbf{d}}
\def\dvt{\widetilde{\dv}}

\def\E{\mathsfbf{E}}

\def\etr{\mathrm{etr}}
\def\F{\mathcal{F}}

\def\gsol{\check{\gamma}}
\def\Gt{\tilde{G}}
\def\gt{\tilde{g}}
\def\H{\mathsfbf{H}}
\def\Hb{\bar{\mathbf{H}}}
\def\Hbh{\widehat{\Hm}}
\def\Hbt{\widetilde{\Hm}}
\def\Hm{\mathbf{H}}

\def\Id{\mathbf{I}}

\def\j{\mathrm{j}\,}

\def\Lam{\boldsymbol{\Lambda}}

\def\muv{\boldsymbol{\mu}}

\def\Mm{\mathbf{M}}

\def\N{\mathcal{N}}
\def\Nm{\mathbf{N}}
\def\nr{{n_{\mathsf{R}}}}
\def\nt{{n_{\mathsf{T}}}}

\def\psh{\hat\psi}

\def\Pm{\mathbf{P}}
\def\Qb{\bar{\mathbf{Q}}}
\def\Qm{\mathbf{Q}}
\def\R{\mathbb{R}}

\def\Rm{\mathbf{R}}

\def\T{\mathsfbf{T}}
\def\Tm{\mathbf{T}}
\def\Tmt{\widetilde{\Tm}}

\def\tr{\mathop{\mathsfbf{Tr}}}
\def\Um{\mathbf{U}}
\def\Umc{\check{\Um}}
\def\uv{\mathbf{u}}
\def\V{\mathsfbf{Var}}
\def\Vm{\mathbf{V}}
\def\Vmc{\check{\Vm}}
\def\vv{\mathbf{v}}
\def\vec{\mathrm{vec}}
\def\w{\check{w}}
\def\Wm{\mathbf{W}}
\def\Wmt{\widetilde{\Wm}}
\def\x{\otimes}

\def\Xm{\mathbf{X}}

\def\xv{\mathbf{x}}

\def\Ym{\mathbf{Y}}

\def\yv{\mathbf{y}}
\def\z{\check{z}}

\def\Zm{\mathbf{Z}}
\def\Zmt{\widetilde{\Zm}}
\def\0{\mathbf{0}}
\def\zv{\mathbf{z}}

\begin{document}
\title{Asymptotic Mutual Information Statistics of Separately-Correlated
Rician Fading MIMO Channels$^*$}
\author{Giorgio Taricco}
\date{19 December 2007}
\maketitle

\def\thefootnote{$*$}
\footnotetext{
This work is supported by the STREP project No.\ IST-026905 (MASCOT) within the
sixth framework programme of the European Commission.
These results have been presented in part at the IEEE GLOBECOM 2006,
Communication Theory Symposium.}
\def\thefootnote{\arabic{footnote}}

\begin{abstract}
Precise characterization of the mutual information of MIMO systems is required
to assess the throughput of wireless communication channels in the presence of
Rician fading and spatial correlation.
Here, we present an asymptotic approach allowing to approximate the distribution
of the mutual information as a Gaussian distribution in order to provide both
the average achievable rate and the outage probability.
More precisely, the mean and variance of the mutual information of the
separately-correlated {\em Rician} fading MIMO channel are derived when the
number of transmit and receive antennas grows asymptotically large and their
ratio approaches a finite constant.
The derivation is based on the {\em replica method}, an asymptotic technique
widely used in theoretical physics and, more recently, in the performance
analysis of communication (CDMA and MIMO) systems.
The replica method allows to analyze very difficult system cases in a
comparatively simple way though some authors pointed out that its assumptions
are not always rigorous.
Being aware of this, we underline the key assumptions made in this setting,
quite similar to the assumptions made in the technical literature using the
replica method in their asymptotic analyses.
As far as concerns the convergence of the mutual information to the Gaussian
distribution, it is shown that it holds under some mild technical conditions,
which are tantamount to assuming that the spatial correlation structure has no
asymptotically dominant eigenmodes.
The accuracy of the asymptotic approach is assessed by providing a sizeable
number of numerical results.
It is shown that the approximation is very accurate in a wide variety of system
settings even when the number of transmit and receive antennas is as small as a
few units.
\end{abstract}

\newpage\section{Introduction}
\markboth{revised for the ieee transactions on information theory}{}

During the last decade, multiple-input multiple-output (MIMO) wireless systems
attracted much interest because of the prediction of outstanding capacity gains
with respect to corresponding single-input single-output (SISO)
systems~\cite{winters,foschini96,telatar}.

Early work in this area was mostly based on the assumption of independent and
identical distributed (iid) Rayleigh fading paths justified by the {\em rich
scattering} assumption~\cite{foschini96}.

However, experimental and theoretical results showed that more realistic MIMO
channel models are required to obtain more accurate results, which can
dramatically impact the potential capacity gains~\cite{gesb02,gold03,sayeed02}.
Thus, more realistic channel models have been proposed in the literature
accounting for both multipath correlation and the presence of a line-of-sight
(LOS) component (Rician fading).

In this work we focus on the separately-correlated MIMO Rician fading channel
whose channel matrix can be written as
\[
  \Hm=\Hb+\Rm^{1/2}\Hm_w\Tm^{1/2},
\]
where $\Hb$ is a constant matrix accounting for the LOS component, $\Rm,\Tm$ are
the nonnegative definite receive and transmit correlation matrices,
respectively, and $\Hm_w$ is a matrix of iid circularly-symmetric zero-mean
complex Gaussian random variables.
We assume that the channel parameters $(\Hb,\Rm,\Tm)$ are unknown at the
transmitter whereas the receiver knows the channel matrix realization exactly
(i.e., following the terminology introduced by Goldsmith {\em et al.}~\cite{gold03},
there is perfect channel-state information at the receiver (CSIR) and no channel
distribution information at the transmitter (CDIT)).

In the MIMO Rician fading channel literature, many works focus on the
{\em spatially uncorrelated} channel.
In particular, H\"osli {\em et al.}~\cite{hoekimlap05} show that the mutual
information is monotonically non-decreasing in the singular values of $\Hb$.
Jayaweera and Poor~\cite{jaypoor05} derive the exact ergodic capacity
when only the Rice factor $K$ known at the transmitter by using the joint
eigenvalue distribution of the noncentral Wishart matrix, and lower and upper
bounds to the ergodic capacity when full CDIT is available.
Hansen and B\"olcskei~\cite{hanboe04} investigate the case of high SNR with
unit--rank average channel matrix.
Venkatesan {\em et al.}\cite{venka03} derive the capacity-achieving covariance
matrix of the channel, generalizing previous results relevant to the case of
unit rank $\Hb$~\cite{jagold04}.
Kang and Alouini~\cite{kangalou06} obtain the exact ergodic capacity by
calculating the determinant of a matrix whose entries contain confluent
hypergeometric functions of the second kind and the exact capacity variance
in terms of Meijer's G-functions.

Many recent works propose upper and lower bounds to the capacity of the
separately correlated Rician fading MIMO channel.
McKay and Collings~\cite{mckay1,mckay2} derive upper and lower bounds to the
ergodic mutual information, based on the method of zonal polynomials, which are
asymptotically tight as the SNR grows to infinity.
They also provide an asymptotic (in the SNR) expression of the mutual
information variance.
The bounds have been further refined in~\cite{mckay3}.
In the same context, Cui {\em et al.}~\cite{cui05} obtain similar upper and
lower bounds to the ergodic capacity by using a determinant expansion in terms
of minors whose moments are calculated by use of hypergeometric matrix
functions.
In a recent paper, Lebrun {\em et al.}\cite{lebfashasmi06} use asymptotic
bounds to confirm that the mean and variance of the MIMO channel capacity
approach the values of the corresponding underlying scattering channel when
the number of antennas is large and the average channel matrix has unit rank.

Another way to circumvent the difficulties related to the derivation of an exact
expression of the capacity consists of assuming an asymptotically large number
of transmit and receive antennas.

In this contest, Moustakas {\em et al.}~\cite{mousimsen03} show that the
capacity of the correlated Rayleigh fading MIMO channel converges to a Gaussian
random variables whose mean and variance are calculated in closed form.
Their derivation is based on approximating the {\em moment generating function}
(MGF) of the capacity by using the {\em replica method}, a technique originally
introduced in the context of statistical physics and successfully applied to
several communications problems (see,
e.g.,~\cite{muller03,guoverdu05,tanaka02,wenwong06}).

A considerable number of research works in this area are based on the use of
Stieltjes Transform and lead to results compatible to those presented here.
Among them, Dumont {\em et al.}~\cite{dum05a} derived the asymptotic mutual
information mean in the case of single-sided (receiver side)
separately-correlated Rician fading.
Hachem {\em et al.}~\cite{haclounaj06,haclounaj07} derive the asymptotic mutual
information mean in the case of uncorrelated (but not iid) Rician fading.
Hachem {\em et al.}~\cite{hac_pastur06} derive the asymptotic mutual information
mean and variance in the case of separately-correlated {\em Rayleigh} fading and
consider the convergence of the random mutual information to the Gaussian
distribution.

Another approach to show that the asymptotic capacity distribution converges to
the Gaussian distribution is based on the study of the channel matrix singular
values.
Focusing on the correlated Rayleigh fading MIMO channel, Martin and
Ottersten~\cite{martott04} show that, subject to certain conditions on the
correlation between the channel elements, the limiting capacity distribution is
Gaussian.

Recent results in this area are aimed at extending the scenario to the MIMO
communication channel with interference.
This direction was taken originally by Moustakas {\em et al.}
in~\cite{mousimsen03} for the Rayleigh fading channel.
More recently,
Riegler and Taricco~\cite{RiegTar07a} derived the second-order statistics of
the separately-correlated Rician fading MIMO channel mutual information in the
presence of interference.
This work was based on the use of Grassman variables and supermatrix theory
(see references therein).

A further extension of the results presented by this paper is the
derivation of the ergodic capacity achieving covariance matrix.
In this contest, the following works are worth mentioning.
\bit
\item
Dumont {\em et al.} provide in~\cite{dum06} an asymptotic method (based on the
Stieltjes transform) to derive the ergodic capacity achieving covariance matrix
for the separately correlated Rician fading MIMO channel.
The key ingredient is the asymptotic ergodic mutual information formula, which
the authors refer to a preliminary version of~\cite{haclounaj07}, and is
equivalent to the expression derived in this paper (by using a the replica
method instead of Stieltjes transforms).
The authors compare their results against those obtained by Vu and Paulraj, who
used an interior point with barrier optimization method, and showed that the
asymptotic method has a considerable advantage in terms of time efficiency.
\item
Taricco and Riegler provide in~\cite{TarRieg07} an algorithm for the derivation
of the ergodic capacity achieving covariance matrix for the separately
correlated Rician fading MIMO channel {\em with interference}.
This work uses an asymptotic approximation of the ergodic mutual information for
a separately-correlated Rician fading MIMO channel in the presence of
interference.
\eit

Finally, it is worth mentioning the recent result obtained by Riegler and
Taricco~\cite{RiegTar07b} in the area of multiuser MIMO communications.
This result consists in the derivation of the ergodic sum-rate achieving
covariance matrices of a separately correlated Rician fading multiple access
MIMO channel and of the resulting ergodic capacity region.
The paper uses the replica method to derive the ergodic mutual information
expressions required to upper bound the partial sum rates of the multiuser
system.

Summarizing, the main contributions of this paper are listed as follows.
\bit
\item
We consider an asymptotic setting where the number of transmit and receive
antennas grows to infinity, while their ratio approaches a finite value.
A closed-form expression of the asymptotic mean and variance of the mutual
information of the separately-correlated MIMO {\em Rician} fading channel for
arbitrary input signal covariance matrix is derived.
These results depend on a pair of scalar parameters obtained by solving two
polynomial equations.
One of the main technical difficulties encountered, in addition to those
addressed in~\cite{mousimsen03}, is in the need to deal with $2\times2$ block
matrices in the saddlepoint approximation which do not split nicely as in the
Rayleigh fading case~\cite{mousimsen03}.
\item
A detailed study of the cumulant generating function of the mutual information
in the asymptotic regime leads to prove that it converges in distribution to a
Gaussian random variable, provided that the correlation matrices $\Tm$ and $\Rm$
satisfy a mild technical condition.
This condition is tantamount to requiring that the channel has no asymptotically
dominant eigenmode.
\item
From a theoretical point of view, the key issues of uniqueness of the
fixed-point equation solution and of reality of the variance expression are
addressed in two appendices.
The latter result is important also because it allows to say that the
replica-symmetric saddlepoint is stable.
\item
To illustrate the analytic applications of the results obtained, the spatially
uncorrelated Rician fading MIMO channel is considered and its ergodic capacity
achieving covariance matrix is derived by using the asymptotic approximation.
The results obtained are compared to the existing nonasymptotic literature.
It can be noticed that the eigenvectors of the optimum covariance matrix
can be obtained very easily from the average channel matrix (key result
from~\cite{jagold04}).
Moreover, also the eigenvalues are derived --- in the asymptotic setup --- by
using a {\em fixed-point water-filling} algorithm that is much simpler than
the numerical techniques proposed in the literature to solve similar
problems~\cite{jorsboc} (in the nonasymptotic case).
\item
Finally, a set of numerical results is presented to validate the asymptotic
analytic method.
These results depart from a basic system scenario and change one of the
following parameters at a time: number of antennas, antenna ratio, spatial
correlation (summarized in the base of the exponential correlation considered,
common for the transmit and receive sides), Rice factor, and SNR.
As far as concerns the ergodic capacity, it is shown that the relative error
between the asymptotic and exact result (the latter derived by extensive
Monte-Carlo simulation) is always lower than $1\%$, even when the number of
antennas is as low as one or two.
The mutual information standard deviation asymptotic results are less accurate,
with a relative error rising to a few percent units.
However, using the asymptotic Gaussianity to derive an approximate outage mutual
information, the standard deviation error is shown to have a negligible impact
on this channel metric, yielding a relative error lower than $1\%$ in the cases
considered.
\eit
\subsection{Notation and basic results}

We denote (column-) vectors and matrices by lowercase and uppercase boldface
characters, respectively.
The imaginary unit is $\j=\sqrt{-1}$.
The $a$th element of a vector $\xv$ is $(\xv)_a$.
The $(a,b)$th element of a matrix $\Am$ is $(\Am)_{ab}$.
$\delta_{ab}=1$ if $a=b$ and $0$ if $a\ne b$.
The transpose of a matrix $\Am$ is $\Am^\T$.
The Hermitian transpose of a matrix $\Am$ is $\Am^\H$.
The Hermitian part of a square matrix $\Am$ is defined as
$\H(\Am) \triangleq \frac{1}{2}(\Am+\Am^\H)$.
The trace of a matrix is $\tr(\Am)=\sum_a(\Am)_{aa}$.
The exponential of the trace of a matrix is
$\etr(\Am)\triangleq\exp(\tr(\Am))$.
The minimum and maximum eigenvalues of a matrix $\Am$ are
$\lambda_{\min}(\Am)$ and $\lambda_{\max}(\Am)$, respectively.
The minimum and maximum singular values of a matrix $\Am$ are
$\sigma_{\min}(\Am)$ and $\sigma_{\max}(\Am)$, respectively.
The Frobenius norm of a matrix $\Am$ is $\|\Am\|$;
its square can be written as $\|\Am\|^2=\tr(\Am\Am^\H)$.
The Kronecker product of two matrices $\Am$ and $\Bm$ is $\Am\x\Bm$;
If $\Am\in\C^{m_1\times n_1}$ and $\Bm\in\C^{m_2\times n_2}$, then
$\Am\x\Bm\in\C^{m_1m_2\times n_1n_2}$ and we can write it in block matrix form
as $[(\Am)_{ab}\Bm]$ for $a=1,\ldots,m_1$ and $b=1,\ldots,n_1$.
Among the properties of the Kronecker product we recall the
following~\cite[p.\ 475]{hornjohnson}:
$i)~\Am\x(\Bm\x\Cm)=(\Am\x\Bm)\x\Cm$;
$ii)~\prod_k(\Am_k\x\Bm_k)=(\prod_k\Am_k)\x(\prod_k\Bm_k)$;
$iii)~(\Am\x\Bm)^\H=\Am^\H\x\Bm^\H$;
$iv)~(\Am\x\Bm)^{-1}=\Am^{-1}\x\Bm^{-1}$;
$v)~\tr(\Am\x\Bm)=\tr(\Am)\tr(\Bm)$;
$vi)~\det(\Am\x\Bm)=\det(\Am)^n\det(\Bm)^m$ if $\Am\in\C^{m\times m}$
and $\Bm\in\C^{n\times n}$.
$\vec(\Am)$ is the column vector obtained by stacking the columns of $\Am$
on top of each other from left to right.
$\Am^{1/2}$ is the {\em matrix square-root} of the Hermitian nonnegative
definite matrix $\Am$ and is defined as $\Um\Lam^{1/2}\Um^\H$ where
$\Am=\Um\Lam\Um^\H$ is the unitary factorization of
$\Am$~\cite[p.\ 414]{hornjohnson}.
The notation $\xv\sim\N_c(\muv,\Rm)$ defines a vector of complex jointly
circularly-symmetric Gaussian random variables with mean value $\muv=\E[\xv]$ and
covariance matrix $\Rm=\E[\xv\xv^\H]-\muv\muv^\H$ and its joint probability
density function (pdf) is given by
\[
  f(\xv) = \det(\pi\Rm)^{-1}\exp[-(\xv-\muv)^\H\Rm^{-1}(\xv-\muv)]~.
\]
$\E[X]$ and $\V[X]$ are the mean and the variance of the random variable $X$,
respectively.
$(x)_+\triangleq\max\{0,x\}$.
Given the functions $f(x)$ and $g(x)$, we have, for $x\to\infty$,
$f(x)=O(g(x))$ if $\limsup_{x\to\infty}|f(x)/g(x)|<\infty$ and
$f(x)=\Theta(g(x))$ if
\[
  0<\liminf_{x\to\infty}\bigg|\frac{f(x)}{g(x)}\bigg|
  \le\limsup_{x\to\infty}\bigg|\frac{f(x)}{g(x)}\bigg|
  <\infty.
\]

\section{MIMO channel model}

We consider a narrowband block fading channel with $\nt$ transmit and $\nr$
receive antennas characterized by the following equation:
\beq\label{ch.eq}
  \yv = \Hm\xv + \zv.
\eeq
Here,
$\xv\in\C^{\nt \times 1}$ is the transmitted signal vector,
$\Hm\in\C^{\nr \times \nt}$ is the channel matrix,
$\zv\in\C^{\nr \times 1}$ is the additive noise vector, and
$\yv\in\C^{\nr \times 1}$ is the received signal vector.

We assume that the additive noise vector contains iid entries
$(\zv)_a\sim\N_c(0,1)$ for $a=1,\ldots,\nr$.

The channel matrix models separately (or Kronecker) correlated Rician fading
so that it can be written as
\[
  \Hm = \Hb+\Rm^{1/2}\Hm_w\Tm^{1/2}
\]
where $\Hb$ represents the mean value and is related to the presence of a
line-of-sight signal component in the multipath fading channel, the Hermitian
nonnegative definite matrices $\Tm,\Rm$ are the transmit and receive correlation
matrices, and $(\Hm_w)_{ab}\sim\N_c(0,1)$ for $a=1,\ldots,\nr$ and
$b=1,\ldots,\nt$.
The covariance between different entries of $\Hm$ is
\beqa
  \mathop{\rm cov}((\Hm)_{ij},(\Hm)_{i'j'})
  &=&
  \sum_{k,\ell,k',\ell'}\E\Big[
  (\Rm^{1/2})_{ik}(\Hm_w)_{k\ell}(\Tm^{1/2})_{\ell j}
  (\Rm^{1/2})_{i'k'}^*(\Hm_w)_{k'\ell'}^*(\Tm^{1/2})_{\ell'j'}^*\Big]
  \non
  &=&
  \sum_{k,\ell} (\Rm^{1/2})_{ik}(\Tm^{1/2})_{\ell j}
  (\Rm^{1/2})_{i'k}^*(\Tm^{1/2})_{\ell j'}^*
  \non
  &=&
  (\Rm)_{ii'}(\Tm)_{jj'}^*.
  \nonumber
\eeqa

\subsection{Normalizations}

Assume that the input signal vector $\xv$ has zero mean%
\footnote{
There is no point in having a nonzero mean input since the signal power would be
greater but the mutual information would remain the same by the relationship
$I(\xv;\yv)=h(\xv)-h(\xv|\yv)=h(\xv-\muv_x)-h(\xv-\muv_x|\yv)$ where
$\muv_x=\E[\xv]$~\cite[Th.9.6.3]{cover}.}
and covariance matrix $\Qm=\E[\xv\xv^\H]$.
To simplify notation, we define
\[
  \Hbt \triangleq \Hb\Qm^{1/2}
  \qquad\text{and}\qquad
  \Tmt \triangleq \Qm^{1/2}\Tm\Qm^{1/2},
\]
so that the input signal covariance matrix is implicitly accounted for into
$\Hbt$ and $\Tmt$.

Then, the total received power is given by:%
\footnote{
We have:
$\E[(\Hm_w^\H\Rm\Hm_w)_{ij}]
=\sum_{k,\ell}\E[(\Hm_w)_{ki}^*(\Rm)_{k\ell}(\Hm_w)_{\ell j}]
=\tr(\Rm)\delta_{ij}$
}
\beqa\label{rx.power.eq}
  \E[\|\yv\|^2]
  &=&
  \E[\|\Hb\xv\|^2] + \E[\|\Rm^{1/2}\Hm_w\Tm^{1/2}\xv\|^2] + \E[\|\zv\|^2] \non
  &=&
  \tr[\Hb\Qm\Hb^\H]+\tr\{\E[\Hm_w^\H\Rm\Hm_w]\Tm^{1/2}\Qm\Tm^{1/2}\}+\tr[\Id_{\nr}]
  \non
  &=&
  \|\Hbt\|^2+\tr(\Rm)\tr(\Tmt)+\nr~.
\eeqa

According to eq.~\eqref{rx.power.eq}, the channel {\em Rice factor} (defined as
the ratio of the received direct to diffuse power~\cite{farro01}) and the SNR
are given by:
\beq\label{rho.K.constr}
  K = \frac{\|\Hbt\|^2}{\tr(\Rm)\tr(\Tmt)}
  \qquad\text{and}\qquad
  \rho = \frac{(K+1)\tr(\Tmt)\tr(\Rm)}{\nr} .
\eeq
In the special case of scalar input signal covariance matrix, $\Qm=q\Id_{\nt}$,
we have:
\beq\label{rho.K.constr1}
  K = \frac{\|\Hb\|^2}{\tr(\Rm)\tr(\Tm)}
  \qquad\text{and}\qquad
  \Qm = \frac{\rho}{K+1}\frac{\nr}{\tr(\Rm)\tr(\Tm)}\Id_{\nt}.
\eeq
\begin{remark}
The two definitions \eqref{rho.K.constr} and \eqref{rho.K.constr1} of the Rice
factor are different when the input signal covariance matrix is not proportional
to the identity matrix.
Definition \eqref{rho.K.constr} complies with the common knowledge that the
Rice factor represents the ratio between the line-of-sight and the scattered
received signal power.
It has the disadvantage of depending on the input signal covariance matrix,
which may become an issue when the channel capacity is investigated.
In that case, when a constraint on the total input signal power ($\tr(\Qm)$) is
set, one has to resort to definition \eqref{rho.K.constr1} even though it does
not meet the common meaning of the Rice factor.
\end{remark}

\subsection{Asymptotic setting}\label{asy.set.sec}

Here we define the asymptotic setting assumed for the derivation of the mutual
information mean and variance when $\nt,\nr\to\infty$ and $\nt/\nr\to\kappa$,
($0<\kappa<\infty$).

More precisely, we consider the sequence of parameters $\nr=1,2,\ldots$ and
$\nt=\lceil\kappa\nr\rceil$ and a corresponding sequence of deterministic
matrices $\Hb,\Tm,\Rm,\Qm$ and random matrices $\Hm_w$ of suitable dimensions.
Additionally,
\bit
\item
The Rice Factor $K$ and the SNR $\rho$ are constant as $\nr\to\infty$.

\item
The correlation matrices are normalized by
\beq\label{tr.constr}
  \tr(\Tm) = \nt \qquad\text{and}\qquad \tr(\Rm) = \nr.
\eeq

\item
The matrices $\Tmt$ and $\Hbt$ are normalized by
\[
  \tr(\Tmt) = \frac{\rho}{K+1}
  \qquad\text{and}\qquad
  \|\Hbt\|^2 = \frac{K\rho}{K+1}\nr.
\]
\eit

\section{Channel mutual information and cumulant generating function}

It is well known \cite{cover} that the mutual information of a MIMO channel with
channel matrix $\Hm$ and input signal covariance matrix $\Qm$ is given by
\[
  I(\Hm) = \ln\det(\Id_{\nr}+\Hm\Qm\Hm^\H)
  \qquad\text{nats/complex dimension.}
\]
The pdf of this random variable can be derived by its MGF, defined as
\[
  G(\nu) \triangleq \E[\exp(-\nu I(\Hm))] =
  \E\Big[\det(\Id_{\nr}+\Hm\Qm\Hm^\H)^{-\nu}\Big]~.
\]
Our goal is to derive an expression of the {\em cumulant generating
function} (CGF) of $I(\Hm)$, defined as $g(\nu)\triangleq\ln G(\nu)$.%
\footnote{
The cumulant generating function of a real Gaussian random variable with mean
$\mu$ and variance $\sigma^2$ is given by
$g(\nu)=-\nu\mu+\frac{1}{2}\nu^2\sigma^2$.}
This allows us to derive the mean and the variance of $I(\Hm)$ as follows:
\[
  \E[I(\xv;\yv)] = -g'(0^+)
  \qquad\text{and}\qquad
  \V[I(\xv;\yv)] = g''(0^+).
\]
In the following, we resort to the {\em replica method} to obtain an asymptotic
expansion of $g(\nu)$.

\subsection{The replica method}

The replica method was originally used in the study of {\em spin
glasses}~\cite{edwand75}.
Later, it found application in other research areas, such as neural networks,
coding, image processing, and communications~\cite{nishi2001}.

Many research works describe the physical meaning of the replica method (see,
e.g.,~\cite{nishi2001,tanaka02,guoverdu05,muller03,wenwong06}).
Here, we confine ourselves to a short review of its basic assumptions.

According to \cite{nishi2001}, the replica method applies to a sequence of
random variables $Z_N$ (typically representing the {\em partition function} of a
physical system), converging in distribution to some $Z_\infty$.
The method is aimed at obtaining the so-called {\em free energy} of the system:
\beq\label{free.eq}
  \F  = \lim_{N\to\infty}-\frac{1}{N}\E[\ln Z_N]
      = \lim_{N\to\infty}-\frac{1}{N}\bigg\{
      \lim_{\nu\to0}\frac{\partial}{\partial\nu}\E[Z_N^\nu]\bigg\}.
\eeq
The method is convenient when direct calculation is impossible (or very hard)
whereas it is relatively easy to calculate the limit
$\lim_{N\to\infty}\E[Z_N^\nu]$ for positive integer $\nu$.
In typical applications, the partition function can be expressed as a
conditional average: $Z_N=\E[Z(\uv;\xv_N)|\xv_N]$.
Here, $\uv$ represents the system {\em microstate} and $\xv_N$ is a set of $N$
independent {\em quenched}%
\footnote{
A physical system is said to be in {\em quenched disorder} when some random
parameters characterizing its behavior do not evolve in time and are then said
{\em quenched} or frozen.
Spin glasses are a typical example.
Quenched disorder is in contrast to {\em annealed disorder} where all random
parameters evolve in time~\cite{fischerhertz,nishi2001}.
}
random parameters.
The method is based on the definition of $Z_{a,N}=\E[Z(\uv_a;\xv_N)|\xv_N]$ as
{\em replicas} of $Z_N$, obtained by averaging with respect to the independent
microstates $\uv_a$ for $a=1,\ldots,\nu$, conditionally on the quenched
parameters.
Hence, we can write:
\[
  G(\nu,N) \triangleq \E[Z_N^\nu] = \E\bigg[\prod_{a=1}^\nu Z_{a,N}\bigg].
\]
The validity of the replica method is subject to verification of the
following assumptions:
\ben
\item{\em(Extension from positive integers).}
The limit $\Gt(\nu)\triangleq\lim_{N\to\infty}G(\nu,N)$ is a smooth function of
$\nu$ in $\nu=0$.
This function is derived for positive integer $\nu$ and extended to a right
neighborhood of $\nu=0$.
\item{\em(Interchange of limits).}
The limits in \eqref{free.eq} can be exchanged, so that
\[
  \F  = \lim_{\nu\to0}\frac{\partial}{\partial\nu}
  \bigg\{\lim_{N\to\infty}-\frac{1}{N}G(\nu,N)\bigg\}.
\]
\item{\em(Replica symmetry).}
The derivation of the limit
\[
  \lim_{N\to\infty}\E\bigg[\prod_{a=1}^\nu Z(\uv_a;\xv_N)\bigg]
\]
is based on the saddlepoint approximation (see Appendix \ref{saddle.app} for a
brief summary of this topic and \cite{jensen,BH86} for a deeper account) and on
the symmetry of the stationary saddlepoint of $\prod_{a=1}^\nu|Z(\uv_a;\xv_N)|$
with respect to the replicated microstate arguments $\uv_a$.
\een

In our application, the role of the partition function is played by
\[
  Z_\nr=\exp[-I(\Hm)]=\det(\Id_\nr+\Hm\Qm\Hm^\H)^{-1}.
\]
Our goal is to determine the asymptotic series expansion in $\nu$ of
\[
  \gt(\nu) \triangleq \ln\E[Z_\nr^\nu] = \ln\E[\exp(-I(\Hm))],
\]
and to show that only the first and second-degree coefficients survive in the
limit.
This is tantamount to saying that the mutual information converges in
distribution to a Gaussian random variable.

\subsection{Derivation of the cumulant generating function}

Setting $\Am=\Id_\nu,\Cm=\Id_{\nr}+\Hm\Qm\Hm^\H,\Bm=\Dm=\0$ in
eq.~\eqref{ABCD.int.eq} of Appendix~\ref{integrals.app}, we can write:
\beqa
  \gt(\nu)
  &\triangleq&
  \ln\Big\{\E\Big[\det(\Id_{\nr}+\Hm\Qm\Hm^\H)^{-\nu}\Big]\Big\}
  \non
  &=&
  \ln\bigg\{
  \int_{\C^{\nr\times\nu}} \E[\etr\{-\pi\Um^\H(\Id_{\nr}+\Hm\Qm\Hm^\H)\Um\}]d\Um
  \bigg\}.
  \nonumber
\eeqa
Next, by setting $\Am=\Id_\nu,\Cm=\Id_\nt,\Bm=-\Dm=\Qm^{1/2}\Hm^\H\Um$ in
eq.~\eqref{ABCD.int.eq} of Appendix~\ref{integrals.app}, we obtain:
\[
  \etr(-\pi\Um^\H\Hm\Qm\Hm^\H\Um)
  = \int_{\C^{\nt\times\nu}} \etr[-\pi(\Vm^\H\Vm+\Um^\H\Hm\Qm^{1/2}\Vm
  -\Vm^H\Qm^{1/2}\Hm^\H\Um)]d\Vm.
\]
Hence, we can rewrite the CGF as:
\[
  \gt(\nu) = \ln\bigg\{
  \int_{\C^{\nr\times\nu}}d\Um\int_{\C^{\nt\times\nu}}
  \E[\etr\{-\pi(\Um^\H\Um+\Vm^\H\Vm
  +\Um^\H\Hm\Qm^{1/2}\Vm-\Vm^\H\Qm^{1/2}\Hm^\H\Um)\}]
  d\Vm\bigg\}.
\]
Now, we calculate the expectation by observing that
\[
  \tr(\Um^\H\Hm\Qm^{1/2}\Vm)
  = \tr(\Um^\H\Hbt\Vm)+\tr[(\Tmt^{1/2}\Vm\Um^\H\Rm^{1/2})\Hm_w].
\]
We obtain:
\beqa
  \gt(\nu)
  &=&
  \ln\bigg\{\int_{\C^{\nr\times\nu}}d\Um\int_{\C^{\nt\times\nu}}
  \etr\{-\pi(\Um^\H\Um+\Vm^\H\Vm
  +\Um^\H\Hbt\Vm-\Vm^\H\Hbt^\H\Um) \non
  &&
  -\pi^2\Um^\H\Rm\Um\Vm^\H\Tmt\Vm
  \} d\Vm\bigg\} \nonumber
\eeqa
after using the result (obtained by applying, e.g., eq.~\eqref{ABCD.int.eq} of
Appendix~\ref{integrals.app}):
\[
  \E[\etr(\Am\Hm_w-\Hm_w^\H\Am^\H)] = \exp(-\|\Am\|^2).
\]

Finally, by setting $\Am=\pi\Um^\H\Rm\Um$ and $\Bm=\pi\Vm^\H\Tmt\Vm$ in
eq.~\eqref{hub.eq} of Appendix~\ref{integrals.app}, we obtain:
\beqa\label{g.eq}
  \gt(\nu)
  &=&
  \ln\bigg\{
  \int_{\C^{\nr\times\nu}}d\Um\int_{\C^{\nt\times\nu}} d\Vm
  \lim_{\alpha,\beta\downarrow0}
  \int \etr(\Wm\Zm-\alpha\Wm\Wm^\T+\beta\Zm\Zm^\T)
  \non
  &&
  \etr[-\pi(\Um^\H\Um+\Vm^\H\Vm+\Um^\H\Hbt\Vm-\Vm^\H\Hbt^\H\Um
  +\Wm\Um^\H\Rm\Um+\Zm\Vm^\H\Tmt\Vm)] d\mu(\Wm,\Zm)
  \bigg\}
  \non
  &=&
  \ln\bigg\{
  \int \etr(\Wm\Zm) d\mu(\Wm,\Zm)\int_{\C^{\nr\nu}}d\uv\int_{\C^{\nt\nu}}d\vv \non
  &&
  \exp\bigg[
  -\pi(\uv^\H,\vv^\H)
  \begin{pmatrix}
  \Id_{\nu \nr}+\Wm\x\Rm & \Id_\nu\x\Hbt \\
  -\Id_\nu\x\Hbt^\H & \Id_{\nu \nt}+\Zm\x\Tmt
  \end{pmatrix}
  \begin{pmatrix}\uv\\\vv\end{pmatrix}
  \bigg]\bigg\}
  \non
  &=&
  \ln\bigg\{
  \int \etr(\Wm\Zm) {\det}^{-1}\begin{pmatrix}
  \Id_{\nu \nr}+\Wm\x\Rm & \Id_\nu\x\Hbt \\
  -\Id_\nu\x\Hbt^\H & \Id_{\nu \nt}+\Zm\x\Tmt
  \end{pmatrix} d\mu(\Wm,\Zm)
  \bigg\}.
\eeqa
Here, the exchange of the limit for $\alpha,\beta\downarrow0$ and the integral
in $\Um,\Vm$ is allowed from the Dominated Convergence
Theorem~\cite[Th.~5.30]{knapp} since the absolute value of the first integrand,
assuming $\Wm=\Wmt+\j\Wm_0$ and $\Zm=\Zm_0+\j\Zmt$ for real
$\Wmt,\Wm_0,\Zmt,\Zm_0$ determined by the integration contour chosen, is {\em
dominated} by
\[
  \kappa
  \exp\bigg\{-\pi(\|\Um\|^2+\|\Vm\|^2)
  -\alpha\bigg\|\Wmt+\frac{\pi\Um^\H\Rm\Um-\Zm_0}{2\alpha}\bigg\|^2
  -\beta \bigg\|\Zmt+\frac{\pi\Vm^\H\Tmt\Vm+\Wm_0}{2\alpha}\bigg\|^2
  \bigg\},
\]
which is plainly integrable over
$(\Um,\Vm,\Wmt,\Zmt)\in\C^{\nr\times\nu}\times\C^{\nt\times\nu}
\times\R^{\nu\times\nu}\times\R^{\nu\times\nu}$.
Similarly, the exchange of integration order is allowed by Fubini's
Theorem~\cite[Th.~5.47]{knapp} because the first integrand in \eqref{g.eq} is
measurable as its absolute value is integrable~\cite[Prop.~5.53(c)]{knapp}.
Finally, after the change of integration order and the application of the limit,
we used Property \ref{kron.prop.eq} and eq.~\eqref{ABCD.int.eq} of
Appendix~\ref{integrals.app}.

\subsection{Saddlepoint approximation}

In order to derive the asymptotic limit of \eqref{g.eq} as $\nt,\nr\to\infty$,
we use the method of saddlepoint approximation (see Appendix \ref{saddle.app}
for a summary of the main results and \cite{jensen,BH86} for a deeper account).

To this purpose, we look for a stationary point of the integrand's logarithm in
\eqref{g.eq}, namely,
\[
  \phi(\Wm,\Zm) \triangleq \tr(\Wm\Zm)
  -\ln\det\begin{pmatrix}
  \Id_{\nu \nr}+\Wm\x\Rm & \Id_\nu\x\Hbt \\
  -\Id_\nu\x\Hbt^\H & \Id_{\nu \nt}+\Zm\x\Tmt
  \end{pmatrix}.
\]
To this end, we resort to the following expansion:
\[
  \delta[\ln\det(\Xm)]
  =\sum_{k=1}^\infty\frac{(-1)^{k-1}}{k}\tr[(\Xm^{-1}\delta\Xm)^k] .
\]
Then, assuming {\em replica symmetry}~\cite{mousimsen03}, we look for a
stationary point of $\phi(\Wm,\Zm)$ of the form $\Wm=w\Id_\nu,\Zm=z\Id_\nu$
for some positive $w,z$.
We have:
\[
  \phi(w\Id_\nu,z\Id_\nu) =
  \nu\left\{wz-\ln\det\begin{pmatrix}
    \Id_{\nr}+w\Rm & \Hbt \\ -\Hbt^\H & \Id_{\nt}+z\Tmt
  \end{pmatrix}\right\}.
\]
The values of $w$ and $z$ are obtained by setting the first-order variation of
$\phi(\Wm,\Zm)$ equal to zero.
More generally, we can write the total variation at $(w\Id_\nu,z\Id_\nu)$ as:
\beqa\label{totexp.eq}
  \delta\phi(w\Id_\nu,z\Id_\nu)
  &=&
  \tr(w\delta\Zm+z\delta\Wm+\delta\Wm\delta\Zm)
  \non
  &&
  +\sum_{k=1}^\infty\frac{(-1)^k}{k}\tr\left\{
  \begin{pmatrix}
    \delta\Wm\x(\Am_1\Rm) & \delta\Zm\x(\Bm_1\Tmt) \\
    \delta\Wm\x(\Cm_1\Rm) & \delta\Zm\x(\Dm_1\Tmt)
  \end{pmatrix}^k
  \right\}
\eeqa
where~\cite[p.\ 18]{hornjohnson}:
\beq\label{ABCD.eq}\left\{\begin{array}{lll}
  \Am_1 &=&  [\Id_{\nr}+w\Rm+\Hbt(\Id_{\nt}+z\Tmt)^{-1}\Hbt^\H]^{-1} \\
  \Bm_1 &=& -(\Id_{\nr}+w\Rm)^{-1}\Hbt\Dm_1 \\
  \Cm_1 &=&  (\Id_{\nt}+z\Tmt)^{-1}\Hbt^\H\Am_1 = -\Bm_1^\H \\
  \Dm_1 &=&  [\Id_{\nt}+z\Tmt+\Hbt^\H(\Id_{\nr}+w\Rm)^{-1}\Hbt]^{-1}
\end{array}\right.\eeq

Next, we focus on the second-order expansion of $\delta\phi(w\Id_\nu,z\Id_\nu)$,
since we show (Theorem~\ref{th.gauss}) that subsequent terms vanish as
$\nr\to\infty$.
We have:
\beqa
  \delta\phi(w\Id_\nu,z\Id_\nu)
  &=&
  [z-\tr(\Am_1\Rm)]\tr(\delta\Wm)+[w-\tr(\Dm_1\Tmt)]\tr(\delta\Zm) \non
  &&
  +\frac{1}{2}\tr[(\Am_1\Rm)^2]\tr(\delta\Wm^2)
  +\frac{1}{2}\tr[(\Dm_1\Tmt)^2]\tr(\delta\Zm^2) \non
  &&
  +[1+\tr(\Bm_1\Tmt\Cm_1\Rm)]\tr(\delta\Wm\delta\Zm).
  \nonumber
\eeqa

Thus, the stationary point is characterized by the values of $w,z$ that are
solutions of the following equations:
\beq\label{wz.eq}
\left\{\begin{array}{lll}
  w &=& \tr(\Dm_1\Tmt) = \tr\Big\{
  [z\Id_{\nt}+\Tmt^{-1}+\Tmt^{-1}\Hbt^\H(\Id_{\nr}+w\Rm)^{-1}\Hbt]^{-1}
  \Big\} \\[1mm]
  z &=& \tr(\Am_1\Rm) = \tr\Big\{
  [w\Id_{\nr}+\Rm^{-1}+\Rm^{-1}\Hbt(\Id_{\nt}+z\Tmt)^{-1}\Hbt^\H]^{-1}
  \Big\}
\end{array}\right.
\eeq
The uniqueness of the solution of \eqref{wz.eq} is proved in
Appendix~\ref{wz.uniq.app}.
In the following, the solution of \eqref{wz.eq} will be denoted by $(\w,\z)$,
and the matrices defined in \eqref{ABCD.eq} after setting $w=\w,z=\z$ will be
denoted by $\Ac,\Bc,\Cc,\Dc$, respectively.

By the previous results, the saddle-point asymptotic approximation of $g(\nu)$
(calculated along the directions of steepest descent from the stationary point)
can be written as:
\beqa
  g(\nu)
  &\sim&
  \nu\left[\w\z-\ln\det\begin{pmatrix}
    \Id_{\nr}+\w\Rm & \Hbt \\
    -\Hbt^\H        & \Id_{\nt}+\z\Tmt
  \end{pmatrix}\right] \non
  &&+
  \ln\bigg\{\frac{1}{(2\pi)^{\nu^2}}\int_{\R^{\nu\times\nu}}\int_{\R^{\nu\times\nu}}
  \etr\bigg[\frac{1}{2}\Big\{\asol\Xm^2-\bsol\Ym^2+2\,\j\gsol\Xm\Ym\Big\}
  \bigg]d\Xm d\Ym\bigg\}
  \nonumber
\eeqa
where $\asol\triangleq\tr[(\Ac\Rm)^2]$,
$\bsol\triangleq\tr[(\Dc\Tmt)^2]$, and
$\gsol\triangleq1+\tr(\Bc\Tmt\Cc\Rm)$.
Using the trace expansion
\beqa
  \tr[\asol(\Xm)^2-\bsol(\Ym)^2+2\,\j\,\gsol\Xm\Ym]
  &=&
  \sum_a[\asol(\Xm)_{aa}^2-\bsol(\Ym)_{aa}^2
  +2\,\j\,\gsol(\Xm)_{aa}(\Ym)_{aa}] \non
  &&+
  2\sum_{a<b}[\asol(\Xm)_{ab}(\Xm)_{ba}-\bsol(\Ym)_{ab}(\Ym)_{ba}
  +2\,\j\,\gsol(\Xm)_{ab}(\Ym)_{ba}],
  \nonumber
\eeqa
the CGF can be readily evaluated as
\[
  g(\nu) \sim \nu\left[\w\z-\ln\det\begin{pmatrix}
    \Id_{\nr}+\w\Rm & \Hbt \\
    -\Hbt^\H        & \Id_{\nt}+\z\Tmt
  \end{pmatrix}\right]
  -\frac{1}{2}\nu^2\ln(\gsol^2-\asol\bsol).
\]
It can be shown that $0<\gsol^2-\asol\bsol<1$ (see Appendix \ref{variance.app}).
As noticed in \cite[Sec.\ IV-C]{mousim07} in a different setting, this
inequality is required to guarantee the local stability of the replica-symmetric
saddlepoint against variations around it.

Thus, we have the following asymptotic expressions of the mean and variance of
$I(\xv;\yv)$:
\beq\label{mean.var.eq}
  \left\{\begin{array}{rcl}
  \E[I(\xv;\yv)]
  &\sim&
  \mu_I\triangleq
  \dst\ln\det\begin{pmatrix}
    \Id_{\nr}+\w\Rm & \Hbt \\
    -\Hbt^\H        & \Id_{\nt}+\z\Tmt
  \end{pmatrix}-\w\z \\
  \V[I(\xv;\yv)]
  &\sim&
  \sigma_I^2\triangleq
  -\ln(\gsol^2-\asol\bsol)
  \end{array}\right.
\eeq
expressed in (nat/complex dimension) and (nat/complex dimension)$^2$,
respectively.

\section{Asymptotic Gaussianity}

The asymptotic Gaussianity of the mutual information requires that all terms of
order $k>2$ in the series expansion \eqref{totexp.eq} vanish as $\nr\to\infty$.
This is shown under some mild technical conditions expressed in the statement of
the following theorem.

\begin{theorem}\label{th.gauss}
In the asymptotic setting specified in Section \ref{asy.set.sec}, assuming
further that the matrices $\Rm$ and $\Tmt=\Tm^{1/2}\Qm\Tm^{1/2}$ satisfy the
conditions
\[
  \sigma_{\max}(\Rm)=\Theta(1) \qquad\text{and}\qquad
  \sigma_{\max}(\Tmt)=\Theta(\nr^{-1}),
\]
as $\nr\to\infty$, then the mutual information $I(\xv;\yv)$ converges in
distribution to a Gaussian random variable with mean and variance specified in
\eqref{mean.var.eq}.
\end{theorem}
\begin{proof}
See Appendix \ref{th.gauss.app}.
\end{proof}

\begin{remark}
The conditions in Theorem \ref{th.gauss} are tantamount to saying that the
spatial correlation structure has no asymptotically dominant eigenmodes, i.e.,
both $\Rm$ and $\Tmt$ do not have asymptotically dominant eigenvalues.
Otherwise, if $\Rm$ and $\Tmt$ had asymptotically dominant eigenvalues, the
asymptotic nature of the MIMO channel would be drastically changed since only a
small subset of antennas would play an asymptotically dominant role while the
remaining ones would be asymptotically irrelevant.

Considering $\Tmt$ instead of $\Tm$ implies that also the structure of the input
covariance matrix might affect the asymptotic evolution of the MIMO system.
As an example, restricting power allocation to a finite subset of transmit
antennas would be equivalent to nullify the corresponding columns of the channel
matrix $\Hm$ and then would drastically change the asymptotic behavior of the
MIMO channel.
\end{remark}

\begin{remark}
From the property of asymptotic Gaussianity, the outage probability
corresponding to a given rate $R$ can be approximated by
\[
  P_o \approx P(\N(\mu_I,\sigma_I^2)<R)
  = Q\bigg(\frac{\mu_I-R}{\sigma_I}\bigg)
\]
as $N\to\infty$.\footnote{
Here, $Q(x)=P(\N(0,1)>x)=\frac{1}{\sqrt{2\pi}}\int_x^\infty\exp(-u^2/2)du$.}
Then, the outage mutual information is given asymptotically by
\beq\label{outcap.eq}
  I_\epsilon \approx \mu_I-\sigma_I~Q^{-1}(\epsilon)
\eeq
with $\epsilon$ denoting the outage probability.
For example, $I_{10\%}\approx\mu_I-1.28\sigma_I$ and
$I_{1\%}\approx\mu_I-2.33\sigma_I$.
As a result, it would be desirable to have a large mean $\mu_I$ and small
standard deviation $\sigma_I$, the latter being a mere consequence of the fact
that the outage mutual information improves as the cumulative mutual information
distribution transition from $0$ to $1$ becomes steeper.
\end{remark}

\section{Analytic example: Ergodic capacity of the uncorrelated channel}

Here we present an application of the asymptotic analytic results obtained to
the calculation of the ergodic capacity of the spatially uncorrelated Rician
fading MIMO channel.
Assuming $\Rm=\Id_\nr$ and $\Tm=\Id_\nt$, the fixed-point eqs.~\eqref{wz.eq}
become:
\beq\label{wz.eq.uncor}
\left\{\begin{array}{lll}
  w &=& \tr\Big\{[z\Id_\nt+\Qm^{-1}+\Hb^\H\Hb/(1+w)]^{-1}\Big\} \\
  z &=& \tr\Big\{[(1+w)\Id_\nr+\Hb(z\Id_\nt+\Qm^{-1})^{-1}\Hb^\H]^{-1}\Big\}
\end{array}\right.
\eeq
After some linear algebra, we can write the asymptotic ergodic mutual
information as follows:
\beq\label{asy.mean.uncor}
  \mu_I = \nr\ln(1+w)+\ln\det
  \bigg[\Id_\nt+\bigg(\Id_\nt+\frac{\Lam_H}{1+w}\bigg)\Qb\bigg]-wz
\eeq
where $\Lam_H$ derives from $\Hb^\H\Hb=\Um_H^\H\Lam_H\Um_H$ and
$\Qb\triangleq\Um_H\Qm\Um_H^\H$.

From Hadamard's inequality~\cite[Th.~7.8.1]{hornjohnson}, we know that the
maximum determinant of a positive definite diagonal matrix is upper bounded by
the product of its diagonal elements and equality holds if and only if the
matrix is diagonal.
Thus, under the power constraint $\tr(\Qm)=\tr(\Qb)\le\rho$, we can apply a
standard water-filling argument~\cite{cover} and obtain:
\beq
  (\Qb)_{ii} = \bigg(\xi-\frac{1+w}{1+w+(\Lam_H)_{ii}}\bigg)_+
\eeq
where $(x)_+\triangleq\max\{0,x\}$ and $\xi$ can be obtained by solving
the equation:
\beq\label{wf.eq.uncor}
  \sum_i \bigg(\xi-\frac{1+w}{1+w+(\Lam_H)_{ii}}\bigg)_+ = \rho.
\eeq
Notice that eqs.~\eqref{wz.eq.uncor} and \eqref{wf.eq.uncor} have to be solved
simultaneously since we have a mutual interdependence between $\Qb$ and the pair
$(w,z)$.
In other words, the solution can be derived by implementing a simple iterative
{\em fixed-point water-filling} algorithm.

Finally, we notice that the structure of the asymptotic ergodic capacity
achieving covariance matrix is consistent with \cite{jagold04}, which showed
that the capacity achieving covariance matrix and the matrix $\Hb^\H\Hb$ have
the same eigenvectors, whereas the problem of calculating the eigenvalues
requires numerical optimization techniques such as those used in~\cite{jorsboc}.
It is then clear the advantage of the asymptotic approximation, which allows to
derive analytic results very simply.

\section{Numerical examples}

In this section we consider a baseline system scenario where we assume that the
average channel matrix is given by $(\Hb)_{ij}=1$ for all $i,j$, the spatial
correlation matrices are of exponential type with common base $\alpha$, namely,
$(\Tm)_{ij}=(\Rm)_{ij}=\alpha^{|i-j|}$, there is no CDIT, so that capacity is
achieved by setting $\Qm=q\Id_t$, and the specification is completed by the
following set of parameters:
\beq
\left\{\begin{array}{ll}
  \nr       &= 4 \\
  \nt/\nr   &= 1 \\
  K         &= 10 \text{~dB} \\
  \alpha    &= 0 \\
\end{array}\right..
\eeq
Then, we consider the impact of changing each parameter in turn, with the aim of
illustrating the accuracy of the asymptotic approximation proposed.
For each parameter we plot the mean and standard deviation of the mutual
information by using the asymptotic method (solid lines) and Monte-Carlo
simulation (markers).

\subsection{Impact of the number of antennas}

Figures \ref{new_mean_vs_antennas} and \ref{new_stdev_vs_antennas} describe the
impact of the number of antennas and of the SNR on the asymptotic approximation
accuracy.
It can be noticed that the approximation of the mean (ergodic mutual information)
is always excellent (the maximum relative error being less than 2\%,
corresponding to the case of $\nt=\nr=2$ and SNR $=30$ dB).
The accuracy of the standard deviation is very good when the SNR is $0$ or
$10$~dB (relative error always less than 1\%) but only fairly accurate
when the SNR is $20$ or $30$~dB (maximum relative error around 5\%).
However, it must be noted that the impact on the outage mutual information
approximation \eqref{outcap.eq} is attenuated by the fact that the standard
deviation is considerably smaller than the mean when the SNR is sufficiently
large.

\subsection{Impact of the antenna ratio}

Figures \ref{new_mean_vs_tx2rx_antennas} and \ref{new_stdev_vs_tx2rx_antennas}
describe the impact of the TX to RX antenna ratio and of the SNR on the
asymptotic approximation accuracy.
It is assumed that the number of receive antenna is fixed, $\nr=4$, and the
number of transmit antennas ranges from $\nt=1$ to $10$.
Both the mean and the standard deviation of the mutual information display very
good accuracy with a maximum relative error smaller than 2\%.

\subsection{Impact of the Rice factor}

Figures \ref{new_mean_vs_k} and \ref{new_stdev_vs_k} describe the impact of the
Rice factor and of the SNR on the asymptotic approximation accuracy.
It can be noticed that the ergodic capacity accuracy is always very good in the
cases observed (maximum relative error around 1\%) while the standard deviation
is slightly overestimated for low $K$ (maximum relative error around 4\%) and
underestimated higher $K$ (the threshold depending on the SNR).
This effect can be explained by considering that the low $K$ condition entails a
larger amount of randomness in the channel matrix that increases the asymptotic
variance expression.
The deviation is larger for high SNR, as in the previous cases observed.
Again, the error on the standard deviation has a modest relative impact on the
outage mutual information which is the goal of the variance approximation.

\subsection{Impact of spatial correlation}

Figures \ref{new_mean_vs_alpha} and \ref{new_stdev_vs_alpha} describe the impact
of spatial correlation of the SNR on the asymptotic approximation accuracy.
The relative error remains below $1\%$ for the ergodic capacity and $3\%$
for the standard deviation.
The relative error increases as the SNR increases and correlation decreases.

It is interesting to note that spatial correlation has a minor effect at
moderate SNR, whereas it becomes important as the SNR increases.

\bigskip

Summarizing the numerical results obtained so far we can say that randomness
tends to increase the relative error between the analytic asymptotic
approximation and the actual value (obtained via Monte-Carlo simulation).
So, increasing the SNR and decreasing the correlation or the Rice factor
produces an increase of randomness that entails a larger relative error.
Besides this, it can be noticed that the relative error gets larger as the
antenna ratio gets closer to $1$.
Finally, the relative error decreases as the number of antennas increases but
this is a trivial consequence of the fact that the approximation is asymptotic
in the number of antennas.

\subsection{Outage mutual information}

In order to assess the accuracy of the proposed asymptotic approximation we
consider here the outage mutual information, defined implicitly as follows:
\beq
  I_\epsilon \triangleq \sup_{R>0}\Big\{R~:~P(I(\Hm)<R)<\epsilon\Big\}.
\eeq
We compare the asymptotic analytic approximation \eqref{outcap.eq} against the
exact numerical result obtained by accurate (though lengthy) Monte-Carlo
simulation.
The results are illustrated in Figs.~\ref{outcap_vs_antennas} and
\ref{outcap_vs_alpha}, reporting $I_{10\%}$ (i.e., the $10\%$ outage mutual
information) versus the number of antennas and the exponential correlation base
$\alpha$, respectively.
The maximum relative error is around $1\%$ in the cases considered, with the
maximum corresponding to the case of $2\times2$ MIMO with the largest SNR.
This is comparable with the relative error already found for the ergodic
capacity approximation.

\section{Conclusions}

This paper presented an analytic approach, based on the replica method, allowing
to approximate the statistics of the mutual information for a separately
correlated Rician fading MIMO channel.
We showed that the mutual information statistics approach the Gaussian
distribution as the number of transmit and receive antennas grows large, with
their ratio approaching a finite constant.

More specifically, we saw that the mutual information mean (corresponding to the
ergodic mutual information of the channel) yields a very accurate approximation
of the real value (which was obtained by extensive Monte-Carlo simulation) with
a relative error never larger than a few percent units.
This remarkable accuracy was obtained not only with a large number of antennas
but also in cases when the number of antennas is definitely small, even in the
limiting SISO case!
As a result, the asymptotic analysis becomes a valuable tool to assess the
system performance whenever this has to be done analytically, such as in the
case of covariance optimization with the goal of finding the channel capacity.
This application is illustrated in the paper by considering the uncorrelated
Rician fading channel.

We also showed that the mutual information standard deviation approximation is
very good.
This is important in view of the application of these results to the derivation
of the outage mutual information.
We showed that the outage mutual information, obtained by using the asymptotic
mean and standard deviation, is very accurate in the cases considered, since the
relative error is always lower than 1\%.
This result depends on the asymptotic Gaussianity of the mutual information,
which is proven in the paper.

Summarizing, the analytic approach presented allows to address concisely the
performance of realistic models of MIMO channels (based on separately correlated
Rician fading).
The accuracy obtained is sufficient for most applications and eliminates the
need to run computer intensive simulations or implement numerical methods in
orderto optimize the system performance.
Many applications of this asymptotic method are appearing in the conference
literature.
Among them, it is worth mentioning the optimization of the transmitted signal
covariance, with and without interference, and the derivation of the capacity
region of the multiuser multiple access MIMO channel.
A simple example of the former (relevant to the case of uncorrelated Rician
fading without interference) is given in the paper.

\section{Acknowledgments}

The author would like to thank the Associate Editor and the Reviewers for the
constructive comments and suggestions, which allowed him to improve the quality
of the original submission.
Additionally, the author wishes to acknowledge several useful discussions with
Dr.\ Erwin Riegler.

\newpage\appendix

\subsection{Property of the Kronecker product}\label{kron.app}

\begin{prop}\label{kron.prop.eq}
For any matrices $\Am,\Xm,\Bm,\Ym$ such that the product $\Am\Xm^\H\Bm\Ym$
exists and is a square matrix, we have:
\[
  \tr(\Am\Xm^\H\Bm\Ym) = \vec(\Xm)^\H(\Am^\T\x\Bm)\vec(\Ym).
\]
\end{prop}

\begin{proof}
Assume that $\Xm=(\xv_1,\ldots,\xv_m)$ and $\Ym=(\yv_1,\ldots,\yv_n)$.
Then,
\beqa
  \tr(\Am\Xm^\H\Bm\Ym)
  &=&
  \tr\left[\Am\left(\begin{array}{ccc}
  \xv_1^H\Bm\yv_1&\cdots&\xv_1^H\Bm\yv_n\\
  \vdots&\ddots&\vdots\\
  \xv_m^H\Bm\yv_1&\cdots&\xv_m^H\Bm\yv_n\\
  \end{array}
  \right)\right]\non
  &=&
  \sum_{a,b}(\Am)_{a,b}\xv_b^H\Bm\yv_a
  = \sum_b\sum_a\xv_b^H(\Am)_{a,b}\Bm\yv_a \non
  &=&
  \vec(\Xm)^H (\Am^\T\x\Bm)\vec(\Ym). \nonumber
\eeqa
\end{proof}

\subsection{Integrals}\label{integrals.app}

\begin{intexp}
For every pair of Hermitian positive definite matrices
$\Am\in\C^{m\times m},\Cm\in\C^{n\times n}$, and for any matrices
$\Bm,\Dm\in\C^{m\times n}$, we have:
\beqa\label{ABCD.int.eq}
  \int_{\C^{n\times m}} \etr[-\pi(\Am\Um^\H\Cm\Um+\Bm^\H\Um+\Um^\H\Dm)]d\Um
  &=&
  \det(\Am^\T\x\Cm)^{-1}\etr(\pi\Am^{-1}\Bm\Cm^{-1}\Dm^\H)
  \non
  &=&
  \det(\Am)^{-n}\det(\Cm)^{-m}\etr(\pi\Am^{-1}\Bm\Cm^{-1}\Dm^\H).
  \non
\eeqa
\end{intexp}

\begin{proof}
From Property \ref{kron.prop.eq}, we have:
\[
  \tr(\Am\Um^\H\Cm\Um) = \vec(\Um)^H (\Am^\T\x\Cm)\vec(\Um).
\]
Then, we can write the integral on the lhs of \eqref{ABCD.int.eq} as
\[
  \int_{\C^{mn}} \exp[-\pi(\uv^\H\Pm\uv+\bv^\H\uv+\uv^\H\dv)]d\uv
\]
where $\Pm=\Am^\T\x\Cm$, $\bv=\vec(\Bm)$, and $\dv=\vec(\Dm)$.
Since $\Am$ and $\Cm$ are Hermitian and positive definite by assumption, so are
$\Am^\T$ and $\Pm=\Am^\T\x\Cm$.

Applying the change of variables $\vv=\Pm^{1/2}\uv$, we obtain:
\[
  (\det\Pm)^{-1} \int_{\C^{mn}}
  \exp[-\pi(\vv^\H\vv+\bv^\H\Pm^{-1/2}\vv+\vv^\H\Pm^{-1/2}\dv)]d\vv.
\]
This integral can also be written as
\[
  (\det\Pm)^{-1} \prod_{a=1}^{mn}\int_\C
  \exp[-\pi(|v_a|^2+\bt_a^*v_a+v_a^*\dt_a)]dv_a
\]
where we let $\bvt=\Pm^{-1/2}\bv$ and $\dvt=\Pm^{-1/2}\dv$.
Setting $v_a=v_{ax}+jv_{ay}$, we have
\[
  |v_a|^2+\bt_a^*v_a+v_a^*\dt_a
  = v_{ax}^2+v_{ax}(\bt_a^*+\dt_a)
  + v_{ay}^2+jv_{ay}(\bt_a^*-\dt_a) .
\]
By using repeatedly the standard Gaussian integral
\[
  \int_\R\exp[-\pi(x^2+ax)]dx = \exp(\pi a^2/4), \qquad a\in\C,
\]
we obtain:
\beqa
  \lefteqn{\int_{\C^{mn}}
  \exp[-\pi(\vv^\H\vv+\bv^\H\Pm^{-1/2}\vv+\vv^\H\Pm^{-1/2}\bv)]d\vv}
  \non
  &=&
  \prod_{a=1}^{mn} \Big\{
  \exp[\pi(\bt_a^*+\dt_a)^2/4] \exp[-\pi(\bt_a^*-\dt_a)^2/4]
  \Big\}
  \non
  &=&
  \prod_{a=1}^{mn} \exp(\pi\bt_a^*\dt_a)
  = \exp(\pi\bvt^\H\dvt)
  = \exp(\pi\bv^\H\Pm^{-1}\dv).
  \nonumber
\eeqa
Finally, since $\Pm^{-1}=(\Am^{-1})^\T\x\Cm^{-1}$, applying Property
\ref{kron.prop.eq} we get:
\[
  \vec(\Bm)^\H((\Am^{-1})^\T\x\Cm^{-1})\vec(\Dm)
  = \tr(\Am^{-1}\Bm\Cm^{-1}\Dm^\H),
\]
which in turn leads to \eqref{ABCD.int.eq}.
\end{proof}

\begin{remark}
Ref.~\cite[eq.(124)]{mousimsen03} reports the following incorrect result:
\beq\label{mous.eq}
  \int_\R dr\exp(-rv)\int_{\j\R}\exp(rt-ut)\frac{dt}{2\pi\j} = \exp(-uv)
  \qquad u,v\in\C.
\eeq
The previous result is wrong because the function $\exp(-rv)$ is not a test
function for the $\delta$-sequence
$\delta_A(r-u)=\int_{-A}^{A}\exp(\j(r-u)t)\frac{dt}{2\pi}$
(unless $v$ is purely imaginary) since it is unbounded over the integration
domain of $r$~\cite{kanwal}.
However, if we smoothen the test function by it by a suitable Gaussian factor,
we obtain:
\beq\label{mous.corr.eq}
  \int_{\R+\j r_0} dr \int_{t_0+\j\R}\exp\Big\{-\alpha r^2+\beta t^2+rt-ut-rv\Big\}
  \frac{dt}{2\pi\j} =
  \frac{1}{\sqrt{1+4\alpha\beta}}
  \exp\bigg(\frac{-uv-\alpha u^2+\beta v^2}{1+4\alpha\beta}\bigg).
\eeq
This is a Lebesgue integral provided that $\alpha,\beta>0$ and it holds for any
$r_0,t_0\in\R$.
Therefore, we can replace eq.\ \eqref{mous.eq} by:
\beq\label{mous.corr.eq1}
  \exp(-uv) = \lim_{\alpha,\beta\downarrow0}
  \int_{\R+\j r_0} dr \int_{t_0+\j\R}\exp\Big\{-\alpha r^2+\beta t^2+rt-ut-rv\Big\}
  \frac{dt}{2\pi\j}.
\eeq
The previous result allows us to prove the following proposition.
\end{remark}

\begin{intexp} (Robust Hubbard--Stratonovich transformation)
For every pair of matrices $\Am\in\C^{m\times n},\Bm\in\C^{n\times m}$, we have:
\beq\label{hub.eq}
  \etr(-\Am\Bm) = \lim_{\alpha,\beta\downarrow0}
  \int\etr(-\alpha\Wm\Wm^\T+\beta\Zm\Zm^\T+\Wm\Zm-\Am\Wm-\Bm\Zm)d\mu(\Wm,\Zm)
\eeq
where
$d\mu(\Wm,\Zm) = \prod_{a=1}^n\prod_{b=1}^m\frac{1}{2\pi j}d(\Wm)_{ab}d(\Zm)_{ba}$
and integration is carried out along contours parallel to the real axis for each
$(\Wm)_{ab}$ and to the imaginary axis for each $(\Zm)_{ba}$.
\end{intexp}

\begin{proof}
The proof stems from repeated application of \eqref{mous.corr.eq1}.
\end{proof}

\subsection{Proof of Theorem~\ref{th.gauss}}\label{th.gauss.app}

\begin{proof}
Using the complete expansion in \eqref{totexp.eq}, we can write the CGF as:%
\footnote{
Here we used the series expansion
$(1+x)^{-1}\exp(x-x^2/2)=1+\sum_{k=3}^\infty c_kx^k
=1-\frac{1}{3}x^3+\frac{1}{4}x^4-\frac{1}{5}x^5+\frac{2}{9}x^6
-\frac{19}{84}x^7+\dots$.
}
\beqa
  g(\nu)
  &\sim&
  \nu\left[\w\z-\ln\det\begin{pmatrix}
    \Id_{\nr}+\w\Rm & \Hbt \\
    -\Hbt^\H        & \Id_{\nt}+\z\Tmt
  \end{pmatrix}\right]
  \non
  &&+
  \ln\bigg\{\int\frac{1}{(2\pi)^{\nu^2}}~
  \etr\bigg[\frac{1}{2}\Big\{\asol\Um^2-\bsol\Vm^2+2\,j\,\gsol\Um\Vm\Big\}\bigg]
  \non
  &&
  \cdot\left\{1+\sum_{k=3}^\infty c_k
  \tr\left[
  \begin{pmatrix}
    \Um\x(\Am_1\Rm) & j\Vm\x(\Bm_1\Tmt) \\
    \Um\x(\Cm_1\Rm) & j\Vm\x(\Dm_1\Tmt)
  \end{pmatrix}^k
  \right]\right\}
  d\Um d\Vm\bigg\}.
  \nonumber
\eeqa

Thus, in order to prove the convergence of $g(\nu)$ to the Gaussian CGF, we have
to show that, as $\nr\to\infty$,
\beq\label{eq1}
  \int\etr\bigg[\frac{1}{2}\Big\{\asol\Um^2-\bsol\Vm^2+2\,j\,\gsol\Um\Vm\Big\}\bigg]
  \tr\left[\begin{pmatrix}
    \Um\x(\Am_1\Rm) & j\Vm\x(\Bm_1\Tmt) \\
    \Um\x(\Cm_1\Rm) & j\Vm\x(\Dm_1\Tmt)
  \end{pmatrix}^k\right] d\Um d\Vm \to 0
\eeq
for $k\ge3$.
Now, in order to balance the asymptotic order of the matrix multipliers of $\Um$
and $\Vm$, it is convenient to apply the change of variables $\Um=\nr^{-1/2}\Umc$
and $\Vm=\nr^{1/2}\Vmc$.
Then, the limit \eqref{eq1} becomes:
\beqa\label{eq2}
  \lefteqn{\int\etr\bigg[\frac{1}{2}
  \Big\{\nr^{-1}\asol\Umc^2-\nr\bsol\Vmc^2+2\,j\,\gsol\Umc\Vmc\Big\}\bigg]}
  \non
  &&
  \cdot\tr\left[\begin{pmatrix}
    \Umc\x(\nr^{-1/2}\Am_1\Rm) & j\Vmc\x(\nr^{1/2}\Bm_1\Tmt) \\
    \Umc\x(\nr^{-1/2}\Cm_1\Rm) & j\Vmc\x(\nr^{1/2}\Dm_1\Tmt)
  \end{pmatrix}^k\right] d\Umc d\Vmc \stackrel{\nr\to\infty}{\longrightarrow} 0
\eeqa
and since
\beqa
  \asol &=& \tr[(\Ac\Rm)^2] \le \tr(\Rm^2) = O(\nr), \non
  \bsol &=& \tr[(\Dc\Tmt)^2] \le \tr(\Tmt^2) = O(\nr^{-1}), \nonumber
\eeqa
by the theorem's assumptions, it is plain to see that condition \eqref{eq2} is
equivalent to the limit
\beq\label{term.k.eq}
  \tr\left[\begin{pmatrix}
    \Umc\x(\nr^{-1/2}\Am_1\Rm) & j\Vmc\x(\nr^{1/2}\Bm_1\Tmt) \\
    \Umc\x(\nr^{-1/2}\Cm_1\Rm) & j\Vmc\x(\nr^{1/2}\Dm_1\Tmt)
  \end{pmatrix}^k\right]
  \stackrel{\nr\to\infty}{\longrightarrow} 0
\eeq
To prove \eqref{term.k.eq}, we investigate the asymptotic order of the singular
values of the matrices $\Am_1\Rm$, $\Bm_1\Tmt$, $\Cm_1\Rm$, and $\Dm_1\Tmt$.
We shall use the following linear algebra inequalities~\cite{hornjohnson}:
\[
\begin{array}{rcll}
  \sigma_{\max}(\Mm\Nm)
  &\le&
  \sigma_{\max}(\Mm)\sigma_{\max}(\Nm),
  & \text{for any $\Mm,\Nm$ with compatible sizes}
  \\
  \sigma_{\max}(\Mm)
  &\le&
  \sigma_{\max}(\Nm),
  & \text{for any square $\Mm,\Nm$ with $\Mm\le\Nm$}
  \\
  \tr(\Mm)
  &\le&
  \tr(\Nm),
  & \text{for any square $\Mm,\Nm$ with $\Mm\le\Nm$}
\end{array}
\]
From the asymptotic setting definitions of Section \ref{asy.set.sec} we have
the following results.
\ben
\item
Since
\[
  \Am_1\Rm
  = [\w\Id_{\nr}+\Rm^{-1}+\Rm^{-1}\Hbt(\Id_{\nt}+\z\Tmt)^{-1}\Hbt^\H]^{-1},
\]
we have $\Am_1\Rm\le\Rm$, so that
\[
  \sigma_{\max}(\Am_1\Rm) \le \sigma_{\max}(\Rm) = \Theta(1)
  \qquad\text{and}\qquad
  \z = \tr(\Am_1\Rm) \le \tr(\Rm) = \Theta(\nr).
\]
\item
Since
\[
  \Dm_1\Tmt
  = [\z\Id_{\nt}+\Tmt^{-1}+\Tmt^{-1}\Hbt^\H(\Id_{\nr}+\w\Rm)^{-1}\Hbt]^{-1},
\]
we have $\Dm_1\Tmt\le\Tmt$, so that
\[
  \sigma_{\max}(\Dm_1\Tmt) \le \sigma_{\max}(\Tmt) = \Theta(\nr^{-1})
  \qquad\text{and}\qquad
  \w = \tr(\Dm_1\Tmt) \le \tr(\Tmt) = \Theta(1).
\]
\item
By using the previous inequalities on the maximum singular values we obtain:
\[
  \tr[(\Am_1\Rm)^k] = O(\nr)
  \qquad\text{and}\qquad
  \tr[(\Dm_1\Tmt)^k] = O(\nr^{1-k})
\]
which confirms, for $k=2$, that $\asol=O(\nr)$ and $\bsol=O(\nr^{-1})$.

\item
Setting $\Hbh\triangleq(\Id_\nr+\w\Rm)^{-1/2}\Hbt$, by the definition of $\Bm_1$
in \eqref{ABCD.eq}, we have:
\beqa
  \sigma_{\max}(\Bm_1\Tmt)
  &=&
  \sigma_{\max}[(\Id_\nr+\w\Rm)^{-1/2}\Hbh(\Id_\nt+\z\Tmt+\Hbh^\H\Hbh)^{-1}\Tmt]
  \non
  &\le&
  \sigma_{\max}[(\Id_\nr+\w\Rm)^{-1/2}]
  \sigma_{\max}[\Hbh(\Id_\nt+\z\Tmt+\Hbh^\H\Hbh)^{-1}]
  \sigma_{\max}(\Tmt).
  \nonumber
\eeqa
Since (plainly) $\sigma_{\max}[(\Id_\nr+\w\Rm)^{-1/2}]=O(1)$ and
\beqa
  \sigma_{\max}^2[\Hbh(\Id_\nt+\z\Tmt+\Hbh^\H\Hbh)^{-1}]
  &=&
  \lambda_{\max}[\Hbh^\H\Hbh(\Id_\nt+\z\Tmt+\Hbh^\H\Hbh)^{-2}]
  \non
  &\le&
  \lambda_{\max}[\Hbh^\H\Hbh(\Id_\nt+\Hbh^\H\Hbh)^{-2}]
  \non
  &=&
  \max_i\frac{\lambda_i(\Hbh^\H\Hbh)}{[1+\lambda_i(\Hbh^\H\Hbh)]^2}
  \le \frac{1}{4},
  \nonumber
\eeqa
we have
\[
  \sigma_{\max}(\Bm_1\Tmt) = O(\nr^{-1}).
\]

\item
Setting $\Hbh\triangleq\Hbt(\Id_\nt+\z\Tmt)^{-1/2}$, by the definition of $\Cm_1$
in \eqref{ABCD.eq}, we have:
\beqa
  \sigma_{\max}(\Cm_1\Rm)
  &=&
  \sigma_{\max}[(\Id_\nt+\z\Tmt)^{-1/2}\Hbh^\H(\Id_\nr+\w\Rm+\Hbh\Hbh^\H)^{-1}\Rm]
  \non
  &\le&
  \sigma_{\max}[(\Id_\nt+\z\Tmt)^{-1/2}]
  \sigma_{\max}[\Hbh^\H(\Id_\nr+\w\Rm+\Hbh\Hbh^\H)^{-1}]
  \sigma_{\max}(\Rm).
  \nonumber
\eeqa
Since (plainly) $\sigma_{\max}[(\Id_\nr+\z\Tmt)^{-1/2}]=O(1)$ and
\beqa
  \sigma_{\max}^2[\Hbh^\H(\Id_\nr+\w\Rm+\Hbh\Hbh^\H)^{-1}]
  &=&
  \lambda_{\max}[\Hbh\Hbh^\H(\Id_\nr+\w\Rm+\Hbh\Hbh^\H)^{-2}]
  \non
  &\le&
  \lambda_{\max}[\Hbh\Hbh^\H(\Id_\nr+\Hbh\Hbh^\H)^{-2}]
  \non
  &=&
  \max_i\frac{\lambda_i(\Hbh\Hbh^\H)}{[1+\lambda_i(\Hbh\Hbh^\H)]^2}
  \le \frac{1}{4},
  \nonumber
\eeqa
we have
\[
  \sigma_{\max}(\Cm_1\Rm) = O(1).
\]
\een
In order to prove \eqref{term.k.eq}, notice that the maximum singular values of
the matrices on the rhs of the Kronecker products in
\[
  \begin{pmatrix}
    \Um\x(\nr^{-1/2}\Am_1\Rm) & j\Vm\x(\nr^{1/2}\Bm_1\Tmt) \\
    \Um\x(\nr^{-1/2}\Cm_1\Rm) & j\Vm\x(\nr^{1/2}\Dm_1\Tmt)
  \end{pmatrix}
\]
are all $O(\nr^{-1/2})$ as $\nr\to\infty$.
Thus, if we define
\[
  \begin{pmatrix}
    \phi_{11,k}(\Um,\Vm)\x\Mm_{11,k} & \phi_{12,k}(\Um,\Vm)\x\Mm_{12,k} \\
    \phi_{21,k}(\Um,\Vm)\x\Mm_{21,k} & \phi_{22,k}(\Um,\Vm)\x\Mm_{22,k}
  \end{pmatrix}
  \triangleq
  \begin{pmatrix}
    \Um\x(\nr^{-1/2}\Am_1\Rm) & j\Vm\x(\nr^{1/2}\Bm_1\Tmt) \\
    \Um\x(\nr^{-1/2}\Cm_1\Rm) & j\Vm\x(\nr^{1/2}\Dm_1\Tmt)
  \end{pmatrix}^k
\]
we have $\sigma_{\max}(\Mm_{ij,k})=O(\nr^{-k/2})$.
Therefore, the trace
\[
  \tr(\phi_{11,k}(\Um,\Vm))\tr(\Mm_{11,k})+
  \tr(\phi_{22,k}(\Um,\Vm))\tr(\Mm_{22,k}) = O(\nr^{1-k/2})
\]
approaches $0$ as $\nr\to\infty$ for all $k\ge3$, which proves \eqref{eq1} and
the fact that $g(\nu)$ converges to a Gaussian CGF.

\end{proof}

\subsection{Uniqueness of the solution of eqs.\ \eqref{wz.eq}}
\label{wz.uniq.app}

The proof reported here was inspired by a similar proof included
in~\cite{hac_pastur06}.

\begin{proof}
Eqs.\ \eqref{wz.eq} can be written equivalently as follows:
\beq\label{wz.eq1}
\left\{\begin{array}{lllll}
  1 &=& \psh_1(w,z)
  &\triangleq&
  \dst\frac{\tr(\Dm_1\Tmt)}{w} =
  \tr\Big\{
  [wz\Id_{\nt}+w\Tmt^{-1}+w\Tmt^{-1}\Hbt^\H(\Id_{\nr}+w\Rm)^{-1}\Hbt]^{-1}
  \Big\} \\[3mm]
  1 &=& \psh_2(w,z)
  &\triangleq&
  \dst\frac{\tr(\Am_1\Tmt)}{z} =
  \tr\Big\{
  [wz\Id_{\nr}+z\Rm^{-1}+z\Rm^{-1}\Hbt(\Id_{\nt}+z\Tmt)^{-1}\Hbt^\H]^{-1}
  \Big\}
\end{array}\right.
\eeq
First, we calculate the partial derivatives of these functions with respect to
$w$ and $z$:
\beq\label{ps1.der.eq}
\left\{\begin{array}{lll}
  \dst\frac{\partial\psh_1}{\partial w}
  &=&
  \dst-\frac{\tr\{\Dm_1\Tmt\Dm_1
  [\Id_{\nt}+z\Tmt+\Hbt^\H(\Id_{\nr}+w\Rm)^{-2}\Hbt]\}}{w^2}<0
  \\
  \dst\frac{\partial\psh_1}{\partial z}
  &=&
  \dst-\frac{\tr\{(\Dm_1\Tmt)^2\}}{w}<0
  \\
  \dst\frac{\partial\psh_2}{\partial w}
  &=&
  \dst-\frac{\tr\{(\Am_1\Rm)^2\}}{z}<0
  \\
  \dst\frac{\partial\psh_2}{\partial z}
  &=&
  \dst-\frac{\tr\{\Am_1\Rm\Am_1
  [\Id_{\nr}+w\Rm+\Hbt(\Id_{\nt}+z\Tmt)^{-2}\Hbt^\H]\}}{z^2}<0
\end{array}\right.
\eeq
Then, we notice that $\psh_1(0,z)=\infty$, $\psh_1(\infty,z)=0$, and
$\partial\psh_1/\partial w<0$.
Hence, $\psh_1(w,z)$ is a continuous monotonically decreasing function in $w$
and the equation $\psh_1(w,z)=1$ has a single solution $w=g(z)\in[0,\infty)$,
which is a continuous function of $z$ by the {\em Implicit Function
Theorem}~\cite[Th.\ 3.16]{knapp}.
Moreover, from the first of \eqref{wz.eq1} and from the first two inequalities
of \eqref{ps1.der.eq}, we obtain:
\[
  g'(z) = -\frac{\partial\psh_1/\partial z}{\partial\psh_1/\partial w} < 0.
\]
Then, the uniqueness of the solution of \eqref{wz.eq} requires that the equation
$h(z) \triangleq \psh_2(g(z),z) = 1$ has a single solution.
This result stems from the fact that $\psh_2(g(0),0)=+\infty$,
$\psh_2(g(\infty),\infty)=0$, and
\[
  h'(z) = \bigg[\frac{\partial\psh_2}{\partial w}g'(z)+
  \frac{\partial\psh_2}{\partial z} \bigg]_{w=g(z)} < 0.
\]
The last inequality is equivalent to
\[
  \bigg[
  \frac{\partial\psh_1}{\partial w}\frac{\partial\psh_2}{\partial z} -
  \frac{\partial\psh_1}{\partial z}\frac{\partial\psh_2}{\partial w}
  \bigg]_{w=g(z)}> 0
\]
which can be checked by direct substitution of the relevant expressions.
\end{proof}

\subsection{Saddlepoint approximation}
\label{saddle.app}

Here we report some basic facts about saddlepoint approximation.
A full account on the subject can be found in \cite{jensen}.
In its simplest form, saddlepoint approximation deals with integrals of the
type
\[
  I = \int_\R g(x)\exp[-\lambda h(x)]dx
\]
in the limit for $\lambda\to\infty$.
Assuming $g(x)$ bounded and $h(x)$ with a global minimum at $x=x_0$, the
integral can be approximated, for $\lambda\to\infty$, as:
\[
  I \approx \int_\R g(x_0)\exp[-\lambda h(x_0)-\lambda h''(x_0)(x-x_0)^2/2]dx
  = \sqrt\frac{2\pi}{\lambda h''(x_0)}g(x_0)\exp[-\lambda h(x_0)].
\]
A more refined derivation yields
\[
  I = \frac{\exp[-\lambda h(x_0)]}{\sqrt{\lambda h''(x_0)/(2\pi)}}\bigg\{
  g(x_0)+\frac{g(5h'''-3h''^2h^{iv})-12g'h''h'''+12g''h''^2]}
  {24h''^3}\bigg|_{x=x_0}\frac{1}{\lambda}+O\bigg(\frac{1}{\lambda^2}\bigg)\bigg\}.
\]
This result can be extended to a multidimensional contour integral
\[
  I = \int g(\xv)\exp[-\lambda h(\xv)]d\xv.
\]
Assuming $|g(\xv)|$ bounded and $h(\xv)$ real, smooth, and with a global minimum
at $\xv=\xv_0$, we have, for $\lambda\to\infty$:
\[
  I \approx
  \int g(\xv_0)\exp[-\lambda h(\xv_0)-\lambda(\xv-\xv_0)^\T\Hm_0(\xv-\xv_0)/2]d\xv
  = \frac{g(\xv_0)\exp[-\lambda h(\xv_0)]}{\sqrt{\det[\lambda\Hm_0/(2\pi)]}},
\]
where $\Hm_0$ is the Hessian matrix of $h(\xv)$ at $\xv=\xv_0$, which is assumed
to be positive definite.

\subsection{Proof of the inequalities $0<\gsol^2-\asol\bsol<1$}
\label{variance.app}

\begin{proof}
First, we show that $0<\gsol<1$, so that, since $\asol\bsol>0$, we have
$\gsol^2-\asol\bsol<1$.
\bit
\item
The inequality $\gsol<1$ derives from
\[
  \gsol = 1+\tr(\Bc\Tmt\Cc\Rm)
        = 1-\tr(\Rm^{1/2}\Cc^\H\Tmt\Cc\Rm^{1/2})
        = 1-\|\Tmt^{1/2}\Cc\Rm^{1/2}\|^2.
\]
\item
The inequality $\gsol>0$ derives from
\beqa
  \gsol
  &=&
  1+\tr(\Bc\Tmt\Cc\Rm)
  \non
  &=&
  1-\tr[\Ac\Hbt(\Id_\nt+\z\Tmt)^{-1}\Tmt(\Id_\nt+\z\Tmt)^{-1}\Hbt^\H\Ac\Rm]
  \non
  &\stackrel{(i)}{\ge}&
  1-\z^{-1}\tr[\Ac\Hbt(\Id_\nt+\z\Tmt)^{-1}\Hbt^\H\Ac\Rm]
  \non
  &\stackrel{(ii)}{=}&
  1-\z^{-1}\tr[(\Id_\nr-\Ac(\Id_\nr+\w\Rm))\Ac\Rm]
  \non
  &\stackrel{(iii)}{=}&
  \frac{\tr(\Ac^2\Rm)+\w\asol}{\z},
  \nonumber
\eeqa
from
$(i)$ the matrix inequality $(\Id_\nt+\z\Tmt)^{-1}\Tmt<\z^{-1}\Id_\nt$;
$(ii)$ the definition of $\Ac$; and
$(iii)$ the definition of $\asol=\tr(\Ac\Rm)$.
Now, in order to handle the case of singular $\Rm$, we define the positive
definite matrix $\Rm_\epsilon\triangleq\epsilon\Id_\nr+\Rm$ for $\epsilon>0$.
We have:
\beqa\label{hard.ineq}
  \gsol
  &\ge&
  \frac{\tr(\Ac^2\Rm_\epsilon)-\epsilon\tr(\Ac^2)+\w\asol}{\z}
  \non
  &\stackrel{(i)}{\ge}&
  \frac{\lambda_{\min}(\Rm_\epsilon^{-1})\tr(\Ac\Rm_\epsilon\Ac\Rm_\epsilon)
  -\epsilon\tr(\Ac^2)+\w\asol}{\z}
  \non
  &\stackrel{(ii)}{\to}&
  \frac{\lambda_{\max}(\Rm)^{-1}+\w}{\z}\asol
  \non
  &>& 0,
\eeqa
from
$(i)$ the inequality $\tr(\Um\Vm)\ge\lambda_{\min}(\Um)\tr(\Vm)$ (holding for
nonnegative definite matrices $\Um,\Vm$ as a direct consequence
of~\cite[9.H.1.h]{marshallolkin}); and
$(ii)$ taking the limit for $\epsilon\downarrow0$ of the previous lower bound
and the limit
$\lim_{\epsilon\downarrow0}\lambda_{\min}(\Rm_\epsilon)=\lambda_{\max}(\Rm)^{-1}$.
\eit
The second part of the proof can be accomplished by combining the inequality
\[
  \gsol \ge \frac{\lambda_{\max}(\Rm)^{-1}+\w}{\z}\asol > \frac{\w}{\z}\asol,
\]
deriving from \eqref{hard.ineq}, and
\[
  \gsol \ge \frac{\lambda_{\max}(\Tmt)^{-1}+\z}{\w}\bsol > \frac{\z}{\w}\bsol,
\]
which can be obtained in a similar way.
Multiplying the two inequalities yields $\gsol^2-\asol\bsol>0$.
\end{proof}

\newpage
\def\figsz{0.65}
\insertfig{\figsz}{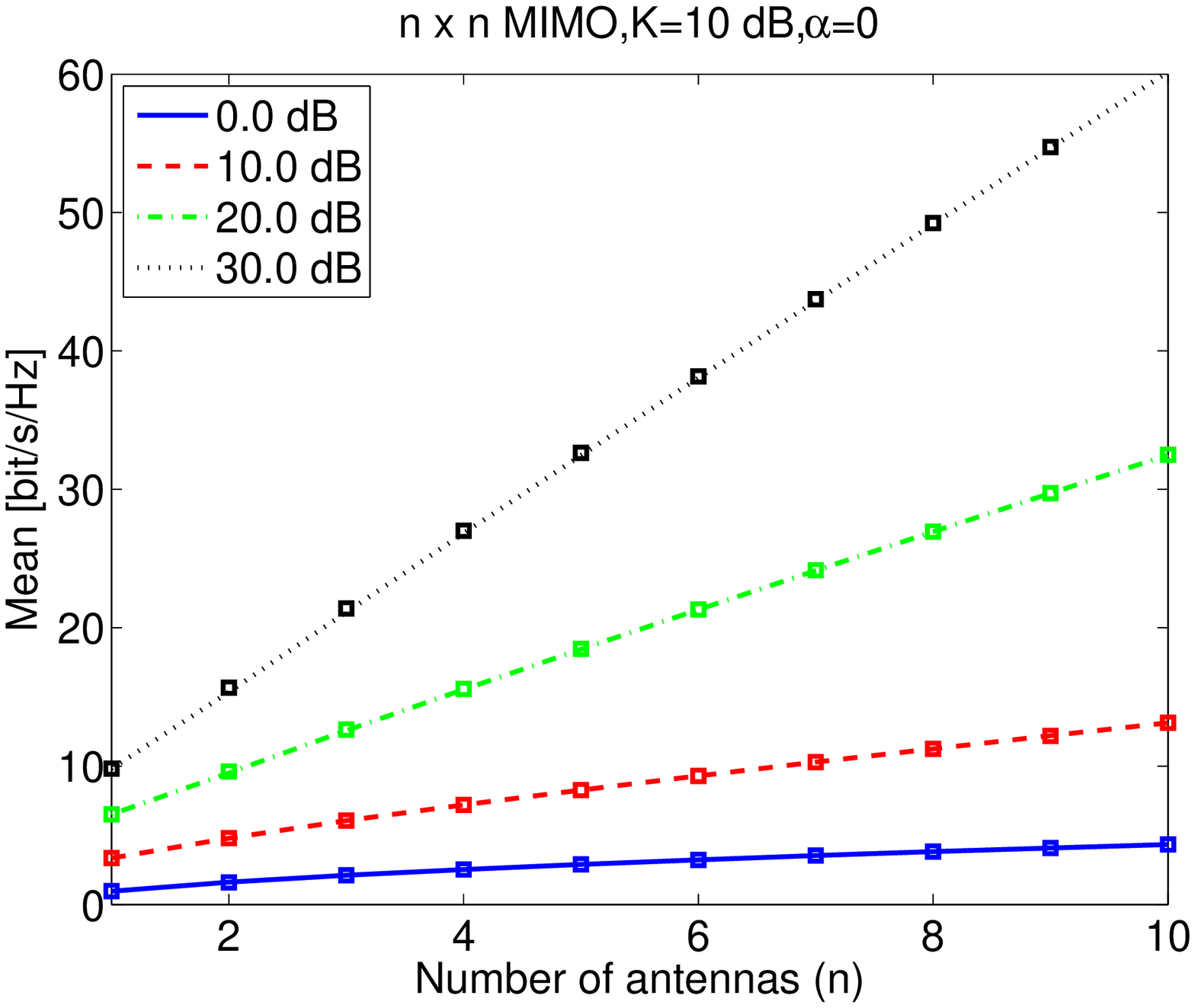}{Capacity mean versus number of antennas for different SNR's.}
\insertfig{\figsz}{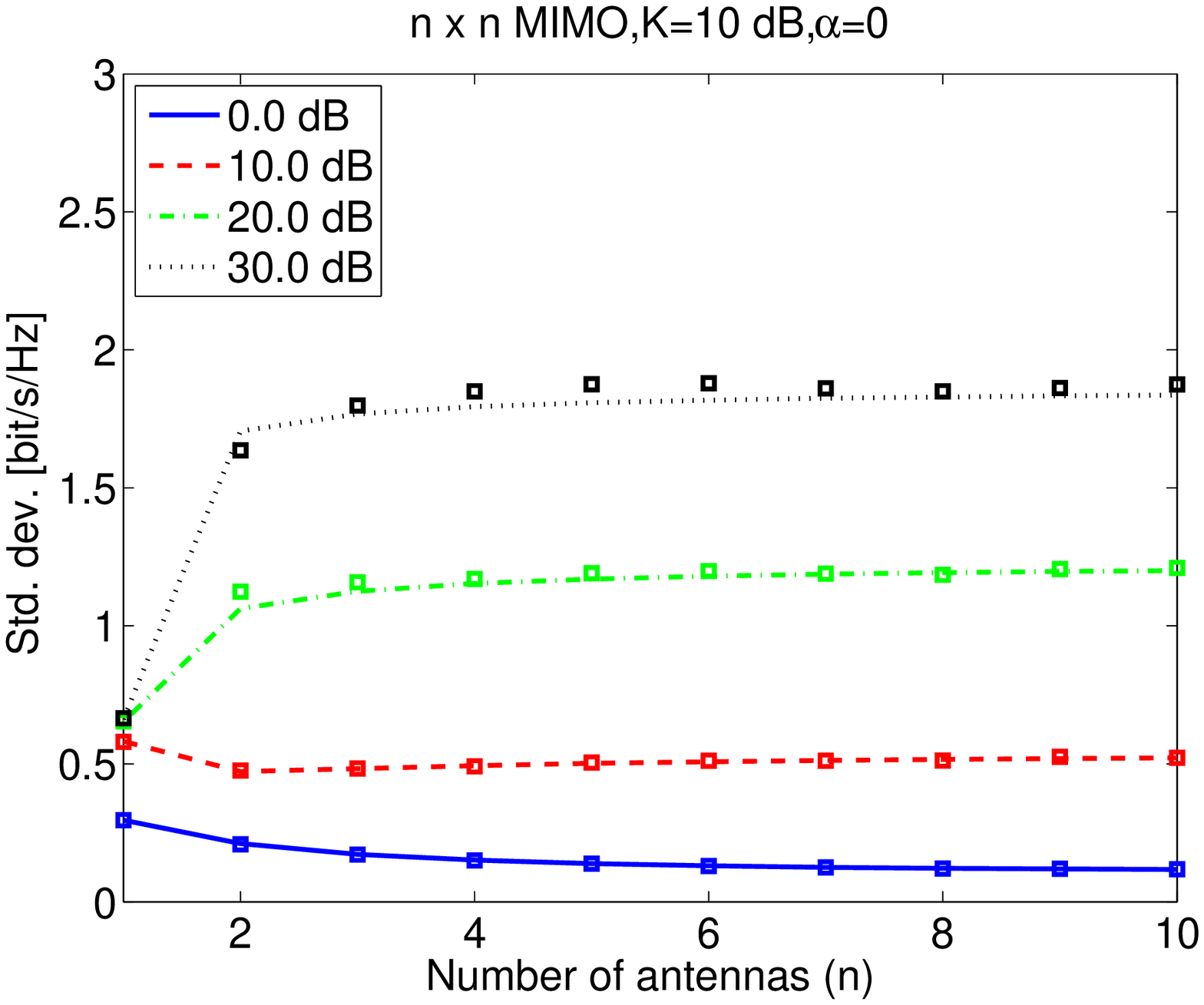}{Capacity standard deviation versus number of antennas for different SNR's.}
\clearpage
\insertfig{\figsz}{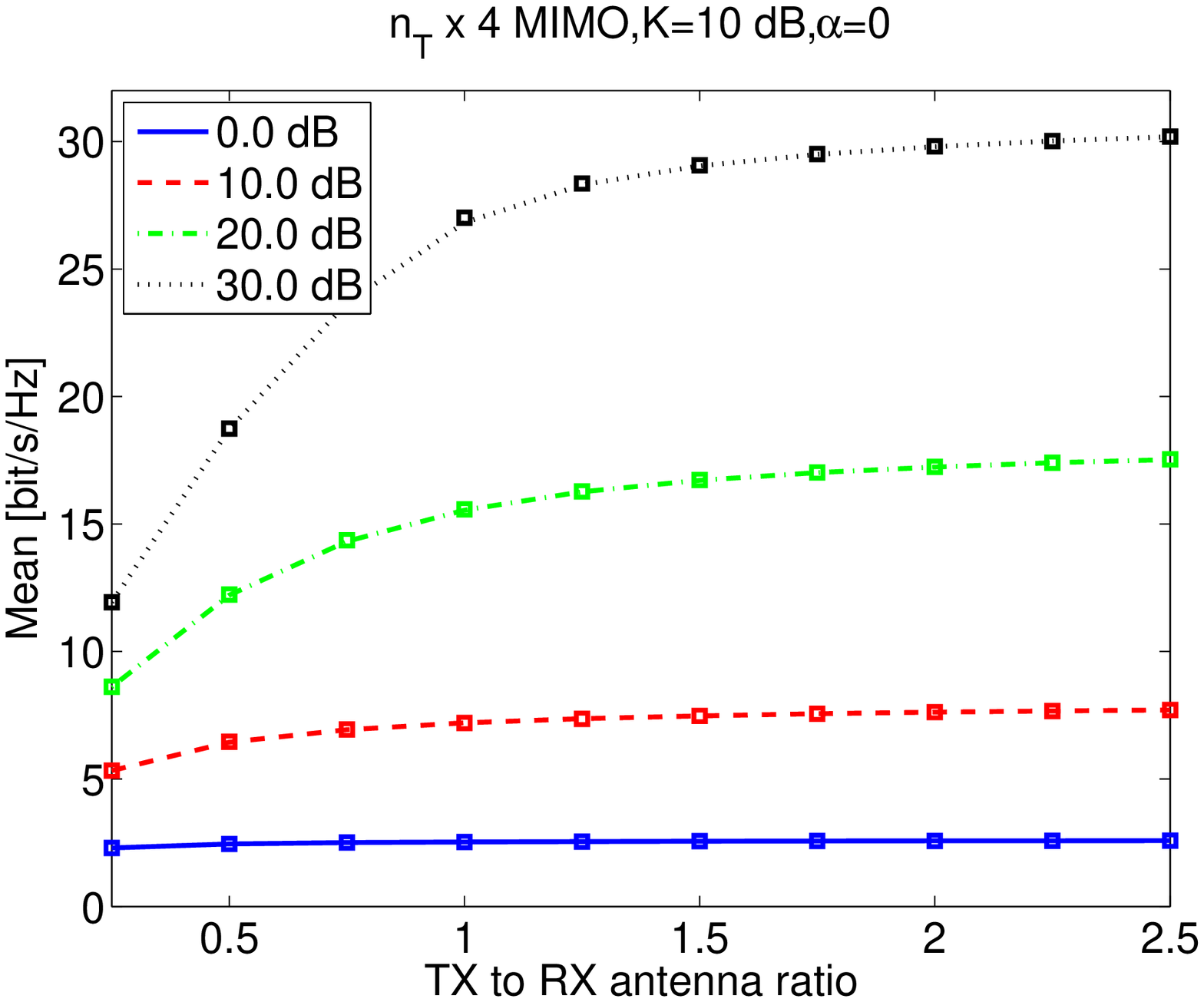}{Capacity mean versus the TX to RX antenna ratio for different SNR's.}
\insertfig{\figsz}{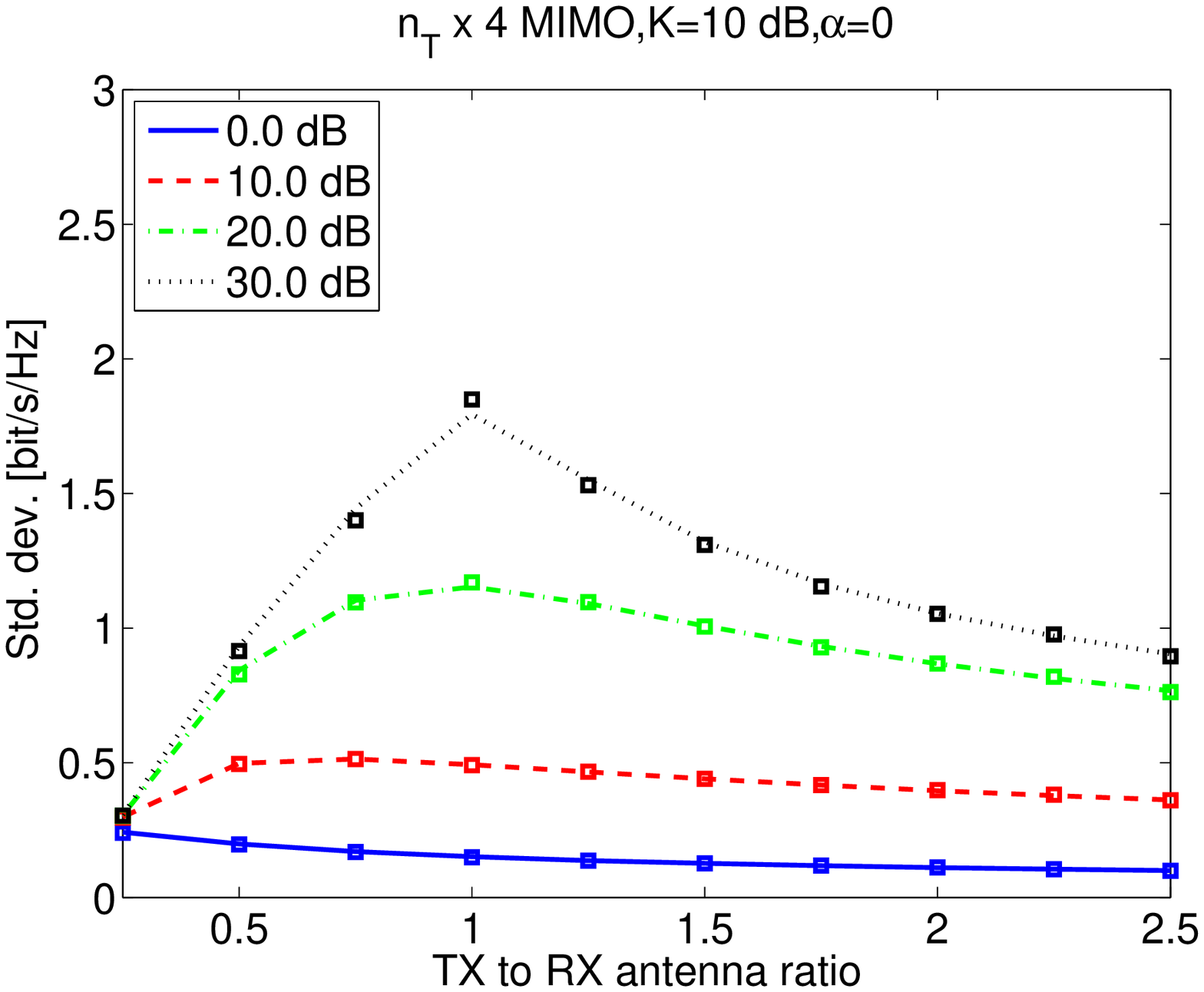}{Capacity standard deviation versus the TX to RX antenna ratio for different SNR's.}
\clearpage
\insertfig{\figsz}{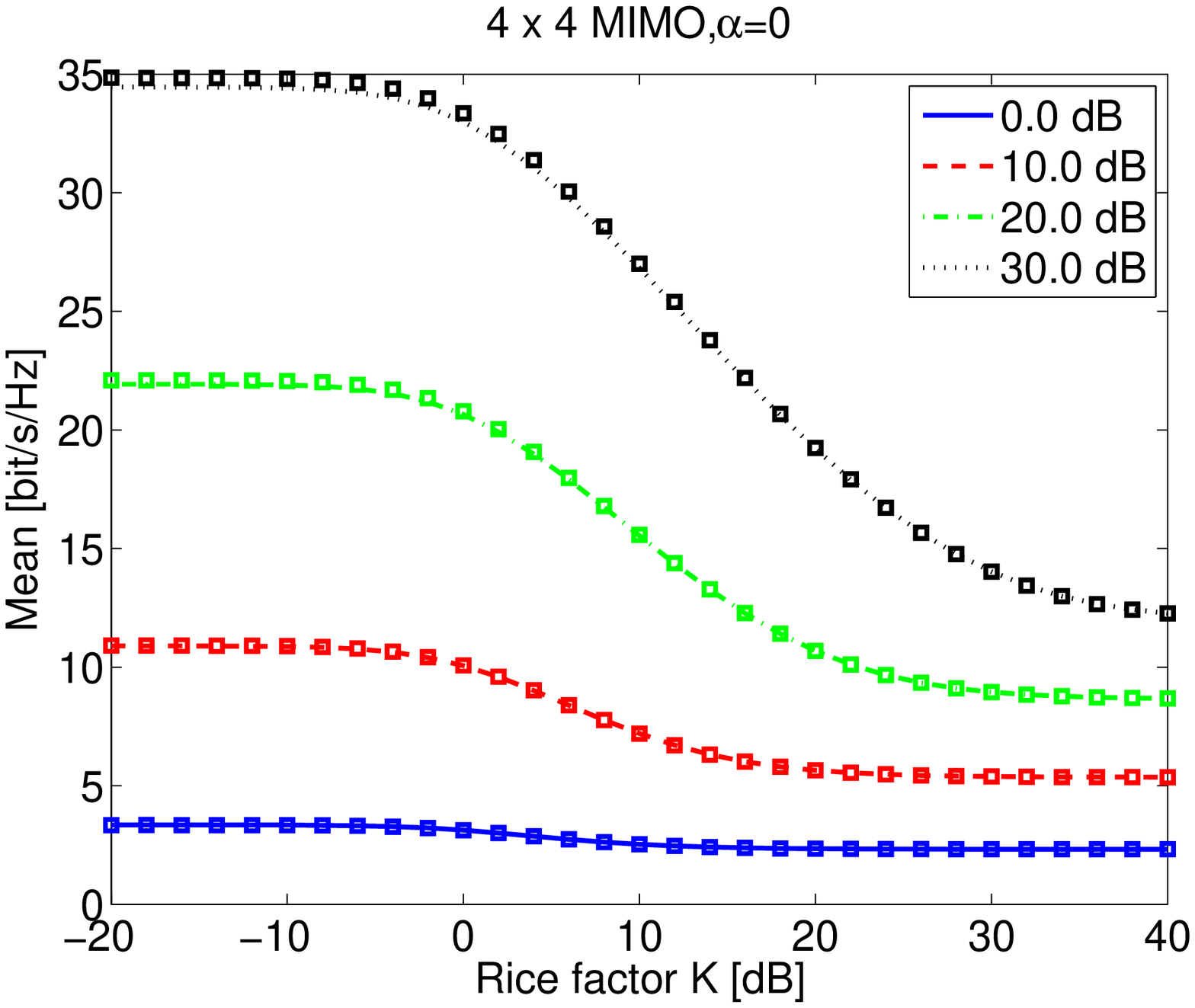}{Capacity mean versus the Rice factor $K$ (expressed in dB) for different SNR's.}
\insertfig{\figsz}{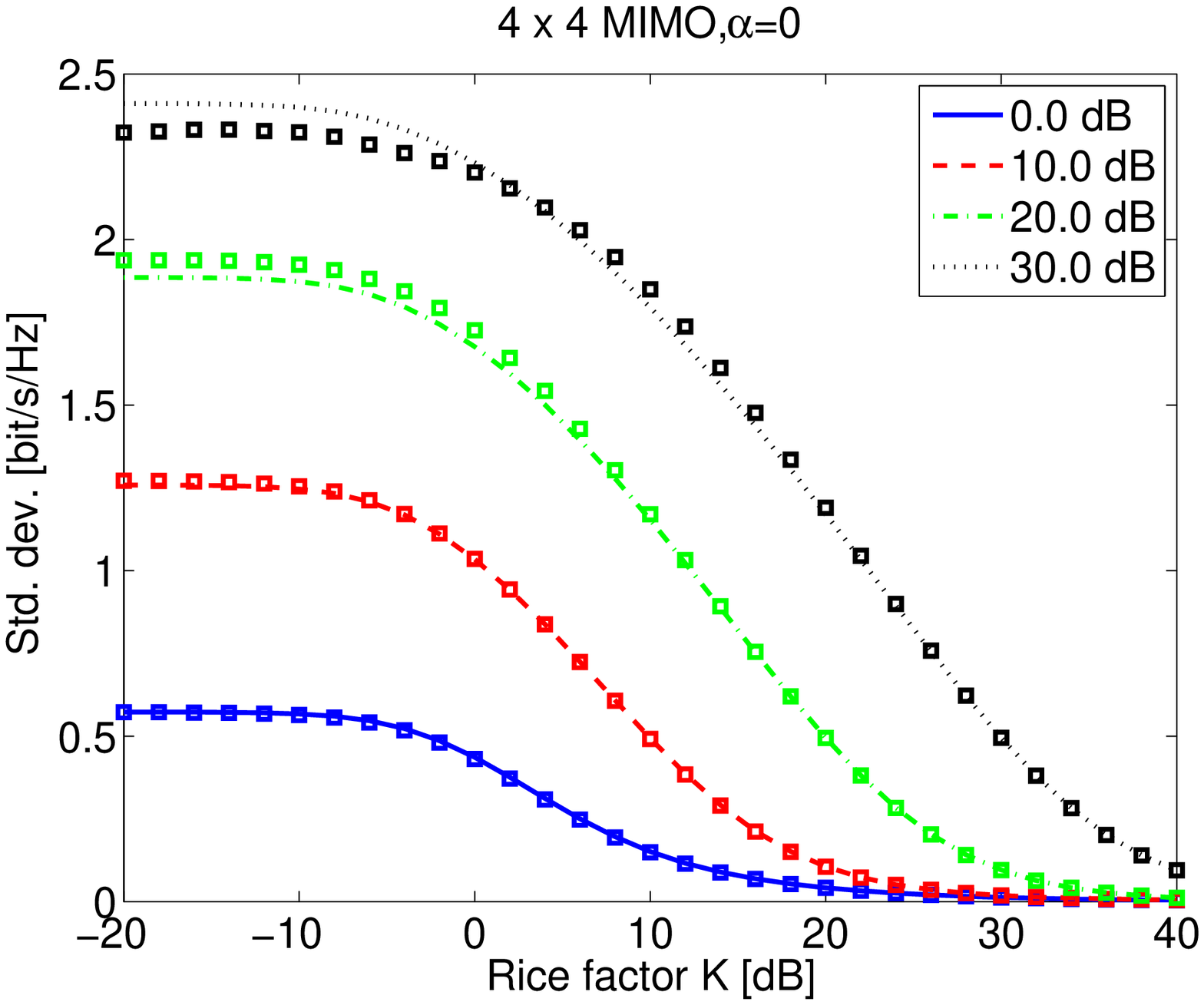}{Capacity standard deviation versus the Rice factor $K$ (expressed in dB) for different SNR's.}
\clearpage
\insertfig{\figsz}{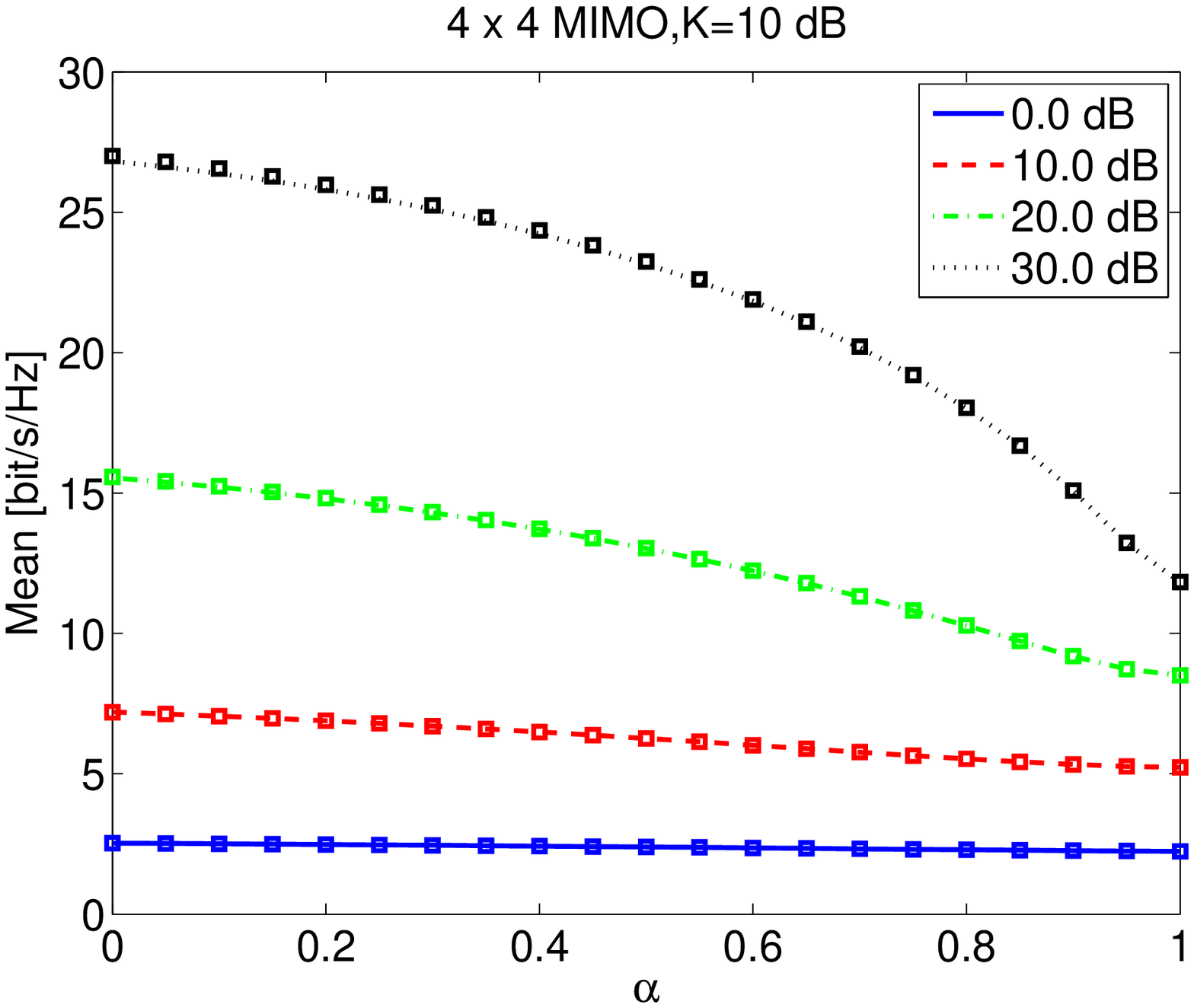}{Capacity mean versus exponential spatial correlation base $\alpha$ for different SNR's.}
\insertfig{\figsz}{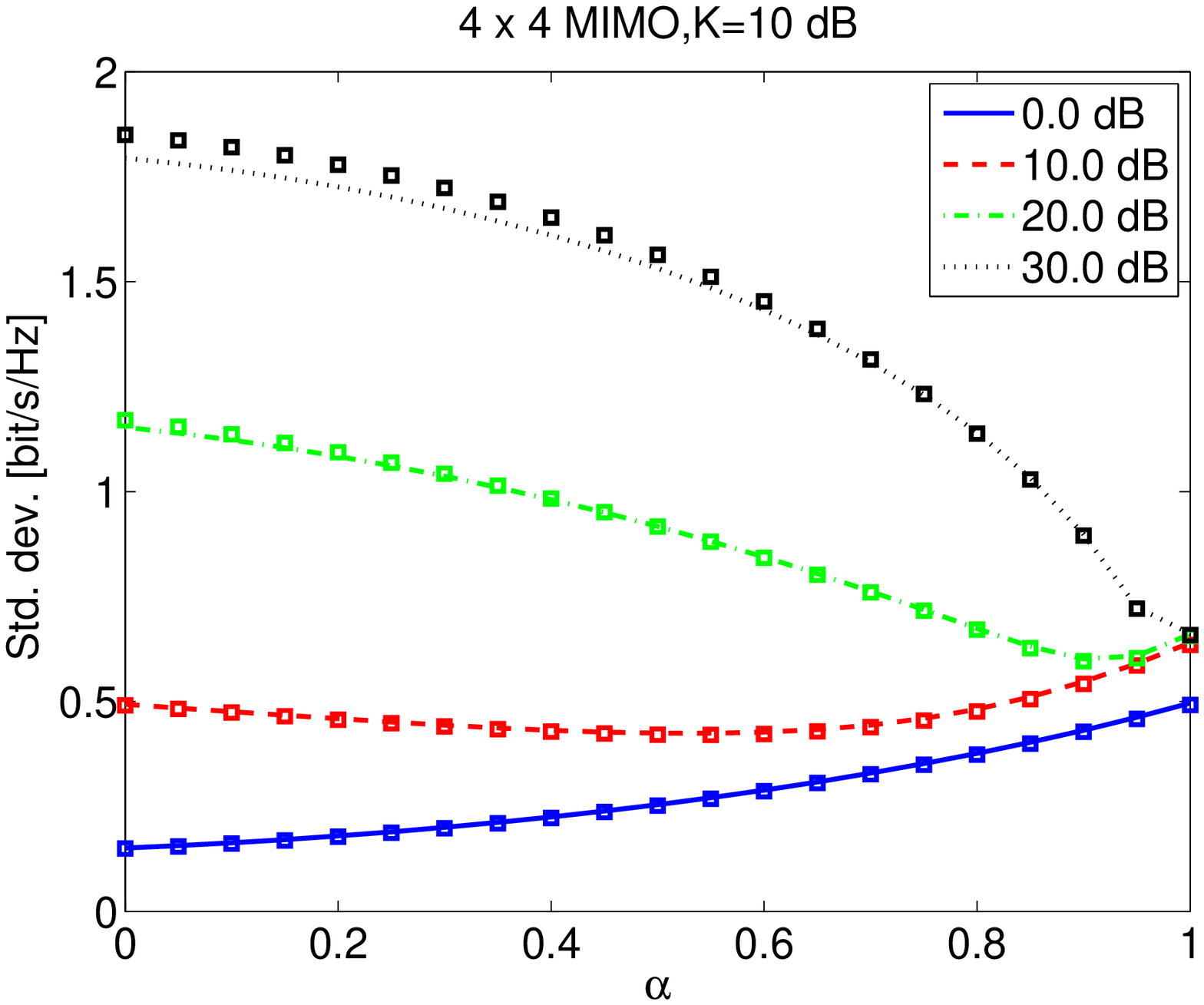}{Capacity standard deviation versus exponential spatial correlation base $\alpha$ for different SNR's.}
\clearpage
\insertfig{\figsz}{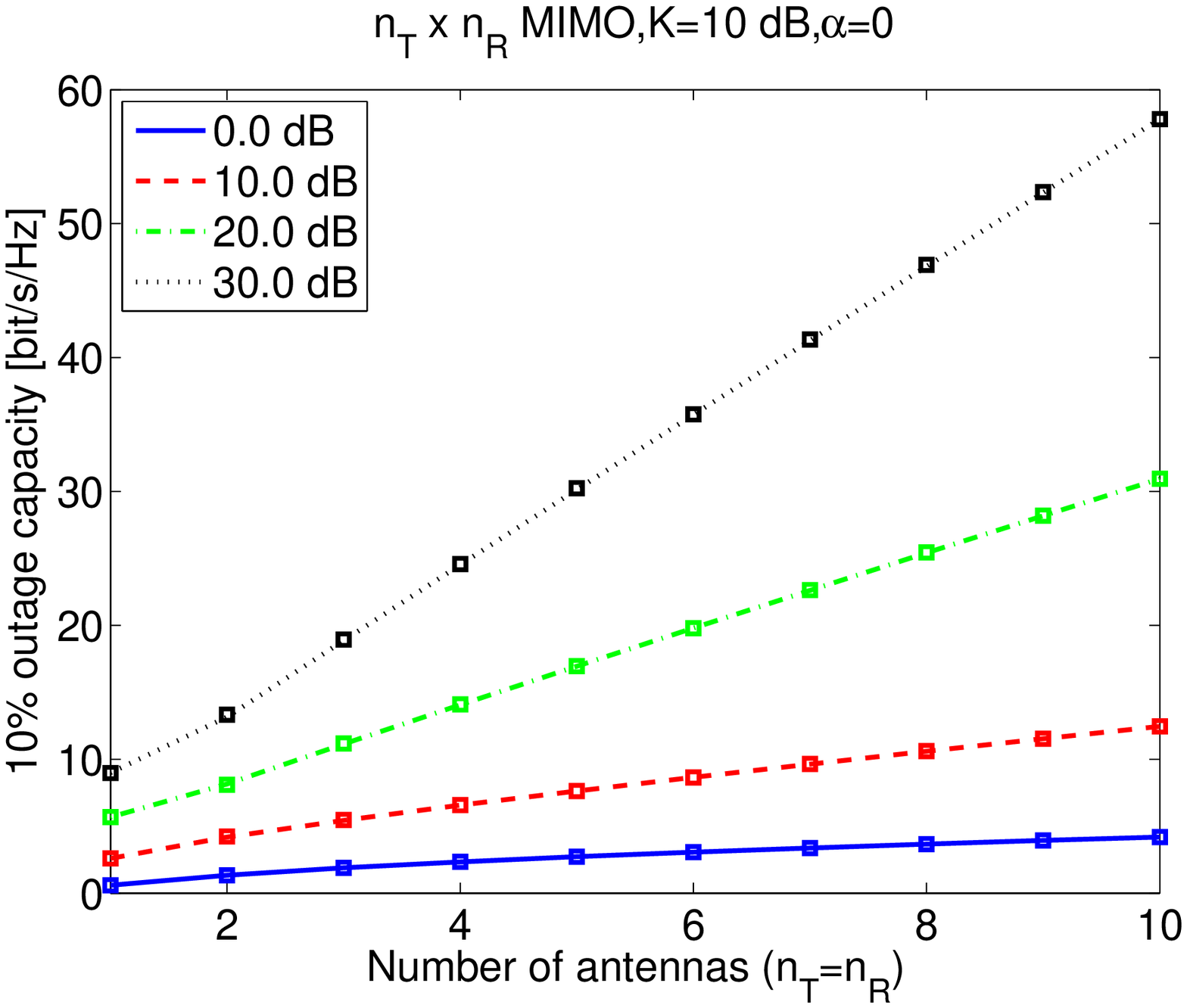}{Outage capacity versus number of antennas for different SNR's.}
\insertfig{\figsz}{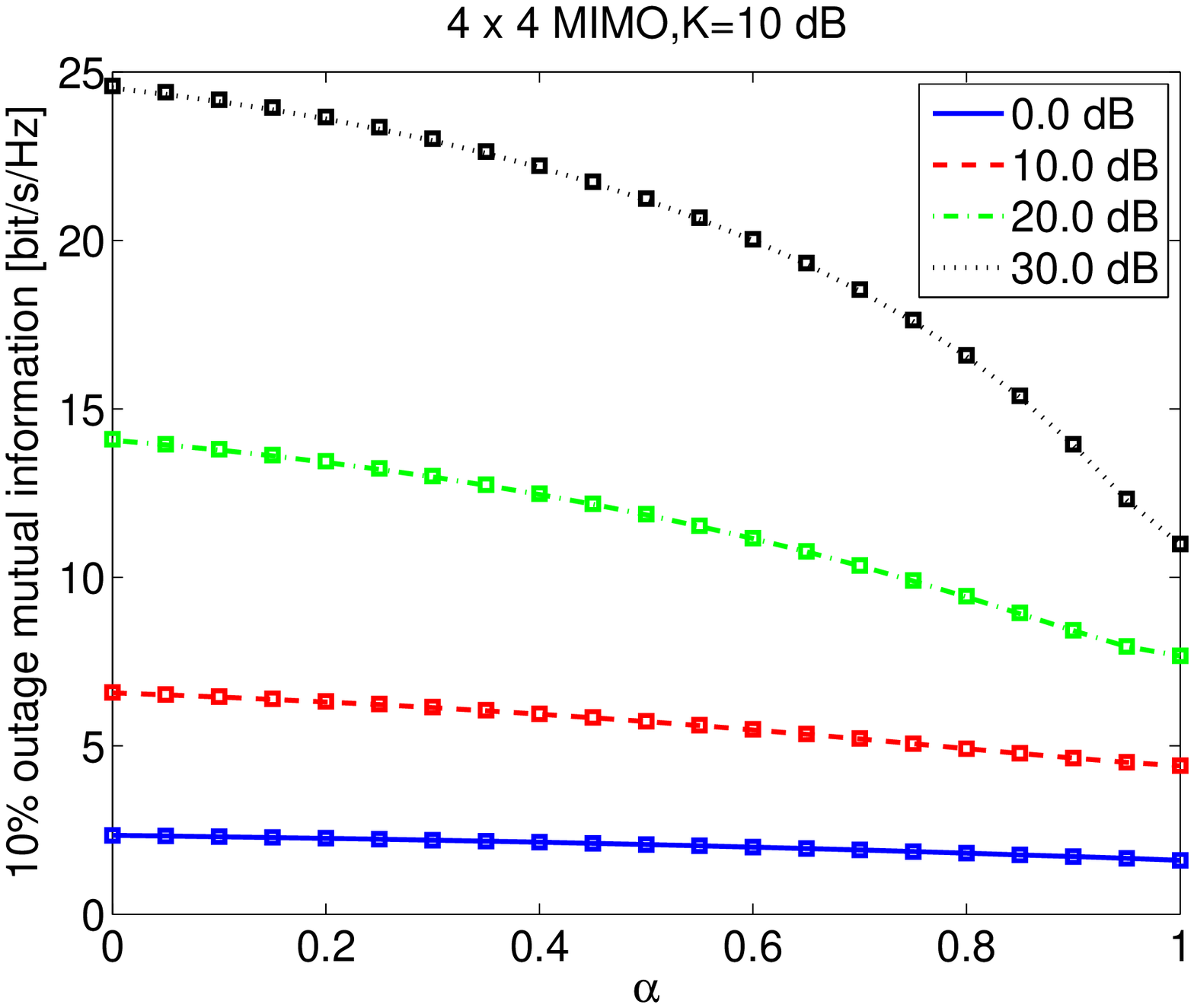}{Outage capacity versus exponential spatial correlation base $\alpha$ for different SNR's.}

\end{document}